\documentclass[aos]{imsart}
\RequirePackage{amsthm,amsfonts,amssymb}
\RequirePackage[numbers,sort]{natbib}
\RequirePackage[colorlinks,citecolor=blue,urlcolor=blue]{hyperref}
\RequirePackage{graphicx}

\usepackage[reqno]{amsmath}

\usepackage{enumerate}
\usepackage{dsfont}
\usepackage[usenames,dvipsnames]{color}
\usepackage{mathrsfs}
\usepackage{placeins}
\usepackage{graphicx}
\usepackage{tikz}
\usetikzlibrary{arrows}
\usepackage{url}
\usepackage{comment}
\usepackage{multicol}
\usepackage{cleveref}
\usepackage{stmaryrd}

\usepackage{natbib}

\renewcommand{\thetable}{\arabic{section}.\arabic{table}}
\numberwithin{equation}{section}

\theoremstyle{plain} 
\newtheorem{theorem}{Theorem}[section]
\newtheorem{lemma}[theorem]{Lemma}
\newtheorem{proposition}[theorem]{Proposition}
\newtheorem{corollary}[theorem]{Corollary}

\theoremstyle{remark}
\newtheorem{definition}[theorem]{Definition}
\newtheorem{dfn}[theorem]{Definition}
\newtheorem{example}[theorem]{Example}

\newtheorem{remark}[theorem]{Remark}

\def\tablename{\footnotesize Table}

\newcommand{\bthe}{\begin{theorem}}
\newcommand{\ethe}{\end{theorem}}

\newcommand{\ben}{\begin{enumerate}}
\newcommand{\een}{\end{enumerate}}

\newcommand{\bit}{\begin{itemize}}
\newcommand{\eit}{\end{itemize}}

\newcommand{\beq}{\begin{equation}}
\newcommand{\eeq}{\end{equation}}

\newcommand{\ble}{\begin{lemma}}
\newcommand{\ele}{\end{lemma}}

\newcommand{\bde}{\begin{definition}\rm}
\newcommand{\ede}{\halmos\end{definition}}

\newcommand{\bco}{\begin{corollary}}
\newcommand{\eco}{\end{corollary}}

\newcommand{\bpr}{\begin{proposition}}
\newcommand{\epr}{\end{proposition}}

\newcommand{\brem}{\begin{remark}\rm}
\newcommand{\erem}{\end{remark}}

\newcommand{\bproof}{\begin{proof}}
\newcommand{\eproof}{\end{proof}}

\newcommand{\bexam}{\begin{example}\rm}
\newcommand{\eexam}{\end{example}}

\newcommand{\bfi}{\begin{fig}}
\newcommand{\efi}{\end{fig}}

\newcommand{\btab}{\begin{tab}}
\newcommand{\etab}{\end{tab}}

\newcommand{\beao}{\begin{eqnarray*}}
\newcommand{\eeao}{\end{eqnarray*}\noindent}

\newcommand{\balo}{\begin{align*}}
\newcommand{\ealo}{\end{align*}}

\newcommand{\balm}{\begin{align}}
\newcommand{\ealm}{\end{align}\noindent}

\newcommand{\beam}{\begin{eqnarray}}
\newcommand{\eeam}{\end{eqnarray}\noindent}

\newcommand{\barr}{\begin{array}}
\newcommand{\earr}{\end{array}}

\newcommand{\C}{\mathbb{C}}

\newcommand{\E}{\mathbb{E}}
\newcommand{\M}{\mathbb{M}}

\renewcommand\P{\mathbb{P}}
\newcommand{\R}{\mathbb{R}}

\newcommand{\VF}[1]{{\color{blue} #1}}

\def\MRV{\mathcal{MRV}}

\def\RV{\mathcal{RV}}

\def\binfty{\boldsymbol \infty}

\def\cB{\mathcal{B}}
\def\cR{\mathcal{R}}

\def\bzero{\boldsymbol 0}
\def\bone{\boldsymbol 1}

\def\bX{\boldsymbol X}
\def\bY{\boldsymbol Y}
\def\bZ{\boldsymbol Z}

\def\bh{\boldsymbol h}
\def\bx{\boldsymbol x}
\def\by{\boldsymbol y}
\def\bz{\boldsymbol z}

\def\be{\boldsymbol e}

\def\bl{\boldsymbol l}

\newcommand{\RVGC}{\text{RVGC}}

\newcommand{\ov}{\overline}

\newcommand{\vague}{\stackrel{\lower0.2ex\hbox{$\scriptscriptstyle
                    \it{v} $}}{\rightarrow}}
\newcommand{\weak}{\stackrel{\lower0.2ex\hbox{$\scriptscriptstyle
                    \it{w} $}}{\rightarrow}}
\newcommand{\what}{\stackrel{\lower0.2ex\hbox{$\scriptscriptstyle
                    \it{\hat{w}} $}}{\rightarrow}}
\newcommand{\eqdis}{\stackrel{\lower0.2ex\hbox{$\scriptscriptstyle
                    \mathrm{d}$}}{=}}
\newcommand{\distr}{\stackrel{\lower0.2ex\hbox{$\scriptscriptstyle
                    \it{d} $}}{\rightarrow}}
\allowdisplaybreaks 
\definecolor{darkgreen}{RGB}{0,139,0}
\newcommand{\DAS}[1]{{\color{red} #1}}

\begin{document}

\begin{frontmatter}
\title{On heavy-tailed risks under Gaussian copula:\\ the effects of marginal transformation}
\runtitle{Heavy-tailed risks and Gaussian copula}

\begin{aug}
  \author[A]{\fnms{Bikramjit} \snm{Das}\ead[label=e1]{bikram@sutd.edu.sg}\orcid{0000-0002-6172-8228}}
    \and
   \author[B]{\fnms{Vicky} \snm{Fasen-Hartmann}\ead[label=e2]{vicky.fasen@kit.edu}\orcid{0000-0002-5758-1999}}
 \address[A]{Engineering Systems and Design, Singapore University of Technology and Design  \printead[presep={,\ }]{e1}}

   \address[B]{Institute of Stochastics, Karlsruhe Institute of Technology\printead[presep={,\ }]{e2}}


  \runauthor{B. Das and  V. Fasen-Hartmann}
\end{aug}

\begin{abstract}
In this paper, we compute multivariate tail risk probabilities where the marginal risks are heavy-tailed and the dependence structure is  a Gaussian copula.  The marginal heavy-tailed risks  are modeled using regular variation which leads to a few interesting consequences.
First, as the threshold increases, we note that the rate of decay of probabilities of tail sets vary depending on the type of tail sets considered and the Gaussian correlation matrix. Second, we discover that although any multivariate model with a Gaussian copula admits the so called asymptotic tail independence property, the joint tail behavior under heavier tailed marginal variables is structurally distinct from that under Gaussian marginal variables.  The results obtained are illustrated using examples and simulations.

\end{abstract}




\end{frontmatter}

\section{Introduction}\label{sec:intro} 

Our  interest is in the asymptotic probabilities of joint tail risk events under Gaussian copula, and its influence on distributions with both heavy and light tailed marginals.
Classically, risk variables have been modeled using the multivariate normal  distribution, nevertheless, in practice, these risks often tend to have tails relatively heavier than normal distributions as observed in many finance, insurance and environmental applications \citep{embrechts_etal:2009,resnickbook:2007}.  The popularity of \emph{copula} modeling in the past few decades has facilitated the use of the Gaussian dependence structure along with any choice of marginal distribution, in particular, heavy-tailed distributions \citep{mcneil-frey-embrechts-05, ibragimov:prokhorov:2017}.

Tail exceedance probabilities have been studied under a variety of model assumptions in the literature. For the popular multivariate normal distributions, such approximations are natural extensions of the Mill's ratio and have a long history; see \cite{moran:1956, dunnett:sobel:1956, savage:1962, elnaggar:mukherjea:1999, dai:mukherjea:2001, hashorva:2005, hashorva:husler:2002, hashorva:husler:2003}. 
In case of heavy-tailed marginal distributions, not necessarily with Gaussian dependence, such tail approximations  have  been extremely useful in insurance and financial risk management, see \citep{embrechts:kluppelberg:mikosch:1997,ibragimov:prokhorov:2017} for details, and \citep{resnickbook:2007}, particularly for the notion of multivariate regular variation which we use to model heavy-tails in this paper. Moreover, we also refer to \citet{asmussen:rojas-nandayapa:2008} for studying threshold exceedance probabilities  with lognormal marginal distributions under Gaussian copula, to \citet{joe:li:2010,hua:joe:li:2014} for using copulas  to model multivariate distributions with regularly varying margins and to \citet{deo:murthy:2021}, for using large deviations to compute tail risks of general loss functions with examples of Gaussian copula dependence. For multivariate copulas, the notion of tail dependence has also been investigated extensively, see \citep{hua:joe:2011,joe:li:nikolouloupoulos:2010,furman:kuznetsov:su:zitikis:2016}.

 Before discussing further, recall that for a random vector $\bX=(X_1,\ldots, X_d)\sim F$ with  continuous marginal distributions $F_1,\ldots, F_d$, the copula $C:[0,1]^d \to [0,1]$ and the survival copula $\widehat{C}:[0,1]^d \to [0,1]$ are {distribution functions} such that $F(\bx)=\P(\bX \le \bx) =C(F_1(x_1), \ldots, F_d(x_d)),$ and $\overline{F}(\bx)=\P(\bX > \bx) =\widehat{C}(\overline{F}_1(x_1), \ldots, \overline{F}_d(x_d)),$ respectively,  for $\bx= (x_1,\ldots,x_d)\in \R^d,$
where $\overline{F}_j=1-F_j\; \forall\, j\in\{1,\ldots,d\}.$ If $\Phi_\Sigma$ defines a $d$-variate  normal distribution with all marginal means zero, variances one and positive semi-definite correlation matrix $\Sigma\in\R^{d\times d}$, and $\Phi$ denotes a standard normal distribution function, then $$C_{\Sigma}(u_1,\ldots,u_d)= \Phi_\Sigma(\Phi^{-1}(u_1),\ldots, \Phi^{-1}(u_d)), \quad \quad 0< u_1,\ldots, u_d < 1,$$
defines a Gaussian copula with correlation matrix $\Sigma$. 
Unless otherwise mentioned, throughout the paper,  $\bZ$
is a Gaussian random vector in $\R^d$ with distribution $\Phi_{\Sigma}$.  In contrast, $\bX$ is a random vector in $\R^d$ with heavy-tailed marginals and Gaussian copula $C_{\Sigma}$.

In \citep{sibuya:1960}, Sibuya showed that if $\bZ=(Z_1, Z_2)$ follows a bivariate normal distribution with any correlation $\rho\in (-1,1)$ where $Z_1, Z_2 \sim \Phi$, then $Z_1$ and $Z_2$ are \emph{asymptotically (tail) independent}, i.e.,
\begin{align}\label{eq:asyind}
    \lim_{t\to\infty}\P(Z_1>t|Z_2>t)= 0.
\end{align}
Asymptotic independence, as in \eqref{eq:asyind}, is equivalent to $\P(Z_1>t, Z_2>t)=o(\P(Z_2>t))$ as $t\to\infty$, meaning that the joint probability of crossings of large thresholds are significantly rarer than individual crossing of the same threshold. Additionally, for a bivariate Gaussian copula with correlation $\rho\in (-1,1)$, the corresponding survival copula $\widehat{C_\rho}$ behaves like  $\widehat C_{\rho}(u,u)= o(u)$ as $u\to0$, which indicates asymptotic independence, irrespective of the marginal distributions; see \citep{ledford:tawn:1996,hua:joe:2011}. For heavy-tailed distributions, the presence or absence of {asymptotic independence}, in various multivariate models have led to an increasing body of research on tail dependence ranging from tail dependence coefficients \citep{ledford:tawn:1996,coles:heffernan:tawn:1999} to hidden regular variation \citep{resnick:2002, mitra:resnick:2011hrv} and more.

In this context of heavy-tailed distributions, Gaussian copula and asymptotic independence, we  provide a reasonable answer to two pertinent issues. 
\begin{enumerate}[(1)]
     \item \emph{Question:} What exact effect does a Gaussian copula and its parameters have on the tail behavior of extreme sets when the marginals are heavy-tailed?\\
     \emph{Answer:}  We establish a precise result on the tail behavior of a random vector $\bX$ with  Gaussian copula and regularly varying tails in \Cref{sec:hrvGaussian}. For a set 
     \begin{align}\label{eq:tailset}
    A= \{\by\in \R^d_+: y_j>x_j \text{ for at least/ exactly $i$ components of $\bx$}\},
\end{align}
        with $\bx=(x_1, \ldots, x_d)\in \R_+^d:=[0,\infty)^d$ and $1\le i \le d$, 
     we are able to approximate
     \[\P(\bX\in tA) \sim t^{-\alpha_*}\ell_*(t)\mu_*(A)\]
    for large $t$ where $\mu^*$ is a non-zero  Borel measure, $\alpha_*>0$ and $\ell_*(tx)/\ell_*(t) \to 1$ (as $t\to \infty$). Both $\mu^*(A)$ and $\ell_*$  depend on the nature of the tail set and the Gaussian correlation matrix; we provide an  explicit representation. To our knowledge this has not yet been characterized in the literature, although heuristic approximations and partial results are mentioned for certain specific sets; see \cite{ledford:tawn:1996,hua:joe:li:2014}.

      \item \emph{Question:} We know that  a Gaussian copula admits asymptotic (tail) independence, but can we claim that as $t\to\infty$,
      \begin{align}\label{eq:genasyind}
          \P(X_1>tx_1, X_2>tx_2)= o(\min_j(\P(X_j>tx_j)))
      \end{align}
      holds as well for any $x_1,x_2>0$, implying that joint threshold crossings are in fact rarer than marginal threshold crossings, even if the thresholds are unequal?\\
      \emph{Answer:}  We observe that under a Gaussian copula dependence, \eqref{eq:genasyind} does hold when the marginals are heavy-tailed, but fails to hold for certain choices of $(x_1, x_2)$ when the marginals are normally distributed, which we discuss further in \Cref{sec:htvslt}. In  \citet{furman:kuznetsov:su:zitikis:2016}, the authors mention that ``\emph{\ldots all classical measures of tail dependence
are such, but they investigate the amount of tail dependence along the main diagonal
of copulas, which has often little in common with the concentration of extremes in the
copulas' domain of definition.}" -- we expand on this for the notions of asymptotic independence and multivariate regular variation with our results.
\end{enumerate}

We note here with caution that the  unrestricted usage of Gaussian copulas \citep{li:2000} in financial risk modeling has led to some serious criticisms \citep{salmon:2009}; see \citep{Donnelly:Embrechts,kluppelberg:resnick:2008,das:embrechts:fasen:2013} for further discussions.
 In this respect, our paper aims to help risk modelers on determining the suitability of a Gaussian copula for their specific tail risk measurements.

 
This paper is structured as follows. In \Cref{sec:prelim}, we present some preliminaries. This includes a discussion on multivariate regular variation on special cones of $\R_+^d$ and certain results on tails of multivariate Gaussian distributions.
The main results of this paper, the asymptotic behavior of probabilities of tail sets of multivariate distributions with tail equivalent regularly varying marginals and Gaussian copula are given in \Cref{sec:hrvGaussian}. Along with a general result, we provide several examples in \Cref{subsec:exmples}. In \Cref{sec:htvslt}, we discuss the structure of tail dependence when comparing heavy-tails and light-tails for bivariate Gaussian copula models. Our results are illustrated using simulations in \Cref{sec:simulation}. Finally, we conclude in \Cref{sec:concl} along with insights on possible risk management applications.

\subsection*{Notations}\label{subsec:notations}
The following notations are used throughout the paper.  The index set is denoted by $\mathbb{I}=\{1,\ldots,d\}$. The cardinality of a set $S\subseteq \mathbb{I}$ is denoted by $|S|$. 
Vector operations are  understood component-wise, e.g., for vectors $\bx=(x_1,\ldots,x_d)$ and $\by=(y_1,\ldots,y_d)$, $\bx\le \by$ means $x_j\le y_j$ for all $j$. 
We also have the following notations for vectors on $\R^d$: $\bzero=\{0,\ldots,0\}^\top, \bone=(1,\ldots,1)^\top,$ $\binfty=(\infty,\ldots,\infty)^\top$ and $e_j =(0,\ldots,1,\ldots,0)^{\top}$, $j\in\mathbb{I}$, where $e_j$ has only one non-zero entry 1 at the $j$-th co-ordinate.  For sequences \linebreak $\bx_t=(x_{t,1},\ldots,x_{t,d}),\by_t=(y_{t,1},\ldots,y_{t,d})\in\R^d$, $t>0$, the notation
$\by_t=\bx_t+o(t)$ as $t\to\infty$ means that $y_{t,j}-x_{t,j}=o(t)$ as $t\to\infty$ for every $j\in\mathbb{I}$.

 For a given vector $\bx\in\R^d$ and $S\subseteq \mathbb{I}$, we denote by  $\bx^{\top}$ the transpose of $\bx$ and by $\bx_S\in\R^{|S|}$ the vector obtained by deleting the components of $\bx$ in $\mathbb{I}\backslash S$. Similarly, for non-empty $I,J\subseteq\mathbb{I}$, $\Sigma_{IJ}$ denotes the appropriate sub-matrix of a given matrix $\Sigma\in\R^{d\times d}$; and we write $\Sigma_I$ for $\Sigma_{II}$.  
  The indicator function of an event $A$ is denoted by $\mathds{1}_{A}$. 
Moreover, for a constant $t>0$ and a set $A\subseteq \R_+^{d}$, we denote by $tA:= \{t\bx: \bx\in A\}$. For a random vector \linebreak $\bX=(X_1,\ldots,X_d)$, we write $\bX\sim F$ if $\bX$ has distribution function $F$; moreover, we understand that marginally $X_j\sim F_j$ for $j\in \mathbb{I}$.
Finally,  $\varphi$ and $\Phi$ are the density  and the distribution function, respectively of a standard normal distribution.

\section{Preliminaries}\label{sec:prelim} 
A few preliminary concepts and results are discussed in this section. In \Cref{subsec:mrv}, we provide a brief introduction to (multivariate) regular variation along with a definition of our model of interest.  In \Cref{subsec:multinormal}, we discuss the asymptotic behavior of probabilities of tail sets of multivariate normal distributions which is used to obtain the main results of this paper.

\subsection{Multivariate regular variation}\label{subsec:mrv}

Our goal is to study tail sets of heavy-tailed models with a Gaussian copula dependence structure. For our heavy-tailed models we assume that the marginal distributions are regularly varying, a popular paradigm for modeling such distributions.

\begin{definition} $\mbox{}$ 
\begin{itemize}
    \item[(a)] A measurable function $f: \R_{+} \to \R_{+}$ is \emph{{regularly varying} (at $+\infty$) with some fixed $\beta\in \R$} if $$\lim_{t\to\infty} f(tx)/f(t) =x^{\beta}, \qquad \forall\,x>0.$$
    We write $f\in\RV_{\beta}$ and if $\beta=0$, then we call $f$ \emph{slowly varying}. 
    \item[(b)] A real-valued random variable $X\sim F$  is regularly varying (at $+\infty$) if the tail 
  $\overline{F}:= 1-F \in \RV_{-\alpha}$ for some $\alpha> 0$. 
\end{itemize}
\end{definition}

The assumption that $X$ is regularly varying at $+\infty$ is equivalent to the existence of a measurable function $b:\R_+\to\R_+$ with $b(t)\to \infty$ as $t\to\infty$ such that
\[t\,\P(X>b(t)x) = t\,\overline{F}(b(t)x) \xrightarrow{t \to \infty} x^{-\alpha}, \quad \forall\, x>0.\]
Consequently, we have $b\in \RV_{1/\alpha}$ and a canonical choice for $b$ is
$$b(t)= F^{\leftarrow}(1-1/t)=\overline{F}^{\leftarrow}(1/t)$$ where $F^{\leftarrow}(x)= \inf\{y\in\R:F(y)\ge x\}$ is the generalized inverse of $F$. Some well-known distributions like Pareto, Burr, L\'evy, Fr\'ech\'et, Student's $t$, are all regularly varying; see \citet{embrechts:kluppelberg:mikosch:1997}.
An elaborate  discussion on regularly varying functions is given in the monograph of \citet{bingham:goldie:teugels:1989} and on regularly varying distributions in \citet{resnickbook:2007,resnickbook:2008}.

Since our particular interest is in  heavy-tailed random vectors $\bX$ in $\R^d$ with a Gaussian copula dependence structure,  we define the following.

\begin{definition}\label{def:rvgc}
An $\R^d$-valued random vector $\bX=(X_1,\ldots,X_d) \sim F$ follows a \emph{regularly varying tailed Gaussian copula distribution} with index $\alpha>0$, scaling function $b$, and correlation matrix $\Sigma$, if the following holds:
\begin{enumerate}[(i)]
    \item The  distribution function $ F_j$ of $X_j$  is continuous, strictly increasing and satisfies \linebreak $\overline{F}_j\in \RV_{-\alpha}$, $\forall\,j\in\mathbb{I}$.
    \item All marginals are completely tail equivalent: $\overline{F}_j(t)/\overline{F}_1(t) \to 1$ as $t\to\infty$, $\forall j\in \mathbb{I}$.
    \item  The function $b:\R_+\to\R_+$ is measurable with $b(t)\sim\ov F_1^{\leftarrow}(1/t)$ as $t\to\infty$. 
    \item The joint distribution function $F$ of $\bX$ is given by  $$F(\bx) = C_{\Sigma}(F_1(x_1),\ldots,F_d(x_d)), \quad \bx=(x_1,\ldots,x_d)\in\R^d,$$
      where $C_{\Sigma}$ denotes the Gaussian copula with correlation matrix $\Sigma\in\R^{d\times d}$.
\end{enumerate}
We write $\bX \in \RVGC(\alpha, b, \Sigma)$ where some of the parameters are ignored if not relevant to the context.
\end{definition}
\begin{remark}
   In the case that the correlation matrix is  an \emph{equi-correlation} matrix given by  $$\Sigma_\rho=\begin{pmatrix}
1 & \;\;\; \rho  & \;\;\; \ldots & \;\;\; \rho\\
\rho & \;\;\; 1  & \;\;\; \ldots & \;\;\; \rho\\
\vdots & \;\;\; \vdots & \;\;\; \vdots & \;\;\; \vdots \\
 \rho & \;\;\; \ldots & \;\;\; \ldots & \;\;\;  1
\end{pmatrix}$$ with $-\frac{1}{d-1}<\rho <1$ (making $\Sigma_\rho$ positive definite),  we write $\bX\in \RVGC(\alpha, b,  \Sigma_{\rho})$. This correlation matrix is used in a few examples later. 
\end{remark} 
Although we assume only marginal regular variation for our random vector $\bX$, eventually we  show that $\bX$ admits \emph{multivariate regular variation} on different subcones of $\R_+^d$.  The subcones we are interested in are
\begin{align*}
\E^{(i)}_{d}  := &\;  \R^d_+ \setminus \{\by\in \R^d_+: y_{(i)}=0\} = \; \{\by\in \R^d_+: y_{(i)}>0\}, \quad\quad 1\le i \le d, \label{def:edd}
\end{align*}
where $y_{(1)}\geq y_{(2)}\geq \ldots\geq y_{(d)}$ is the decreasing order statistic of $y_1,\ldots,y_d$.
 Here $\E^{(1)}_{d}$ is the {positive quadrant} with $\{\bzero\}$ removed, $\E^{(2)}_{d}$ is the {positive quadrant} with all one-dimensional co-ordinate axes removed, $\E^{(3)}_{d}$ is the {positive quadrant} with all two-dimensional co-ordinate hyperplanes removed, and so on.  Clearly,
\begin{equation} \label{eq2.1}
\E^{(1)}_{d} \supset \E^{(2)}_{d} \supset \ldots \supset \E^{(d)}_{d}.
\end{equation}
Note that for any $1\le i \le d$, the set $\{\by\in \R^d_+: y_{(i)}=0\}\subseteq \R_+^d$
is a closed cone in $\R_+^d$ containing $\bzero$. Moreover, for our particular sets of interest of the form $A$ as  in \eqref{eq:tailset}, we have $A\subset \E_d^{(i)}$.

\begin{dfn}\label{def:mrv}
Let $i\in\mathbb{I}$. A random vector $\bY \in \R^{d}$ is \emph{multivariate regularly varying} on $\E^{(i)}_{d}$ if there exists a regularly varying function
$b_i \in \RV_{1/\alpha_i}$, $\alpha_i >0$,  and a non-null (Borel) measure $\mu_i$ which is finite on Borel sets bounded away from $\{\by\in \R^d_+: y_{(i)}=0\}$
{such that} 
\begin{equation*}
    \lim_{t\to\infty}t\,\P\left( \frac{\bY}{b_i(t)} \in A \right)= \mu_i(A)
\end{equation*}
for all Borel sets $A\in\mathcal{B}(\E^{(i)}_{d})$ which are {\it bounded away from\/} $\{\by\in \R^d_+: y_{(i)}=0\}$ with $\mu_i(\partial A)=0$. We write $\bY \in \MRV(\alpha_i, b_i, \mu_i, \E^{(i)}_{d})$; one or more parameters are often dropped according to convenience.
\end{dfn}
The limit measure $\mu_i$ is homogeneous of order $-\alpha_i$, i.e., $\mu_i(\lambda A)=\lambda^{-\alpha_i}\mu_i(A)$ for any $\lambda>0$.
\begin{remark}
The type of convergence used here is $\M$-convergence of measures. It can be discussed in the more general context of multivariate regular variation on subcones \linebreak $\R_+^d\setminus\C_0$ of $\R_+^d$, where $\C_0$ is a closed cone containing $\bzero$. Multivariate regular variation on $\E_d^{(i)}$ is only a special example. More details
can be found in \cite{hult:lindskog:2006a,das:mitra:resnick:2013,lindskog:resnick:roy:2014,das:fasen:2023}.  
\end{remark}

If  $\bY \in \MRV(\alpha_i, b_i, \mu_i, \E^{(i)}_{d})$, $\forall i\in\mathbb{I}$,
then a direct conclusion from \eqref{eq2.1} is that
\beam
    \alpha_1\leq \alpha_2\leq \ldots\leq\alpha_d
\eeam
implying that the rate of decay of probabilities of appropriate  tail sets $tA \subset \E^{(i)}_{d}$ is not as fast as that for tail sets $tB \subset \E^{(i+1)}_{d}$ as $t\to\infty$.
 An alternative and quite useful characterization of regular variation on $\E_d^{(i)}$ follows from \citet[ Proposition 2.7]{das:fasen:2023}. 

\begin{lemma} \label{Lemma 2.3}
Let $\bY$ be a random vector in $\R^d$ and fix $i\in \mathbb{I}$. Suppose $\alpha_i>0$, $b_i\in \RV_{1/\alpha_i}$ is a measurable function and $\mu_i$ is a non-null (Borel) measure  
on $\mathcal{B}(\E_d^{(i)})$ which is finite on Borel sets bounded away from $\{\by\in \R^d_+: y_{(i)}=0\}$.
Then  $\bY \in \MRV(\alpha_i, b_i, \mu_i, \E^{(i)}_{d})$
if and only if
\begin{equation*}\label{eq:RegVarMeas}
    \lim_{t\to\infty}t\,\P\left( \frac{\bY}{b_i(t)} \in A_{\bx_S} \right)= \mu_i(A_{\bx_S})
\end{equation*}
for all sets $S\subseteq \mathbb{I}$ with $|S|\geq i$, for all $x_s>0$ $\forall s\in S$ such that $\bx_S=(x_s)_{s\in S}$ and   
 \begin{align}\label{set:Axs}
        A_{\bx_S} = \{\by\in \R_+^d: y_s > x_s, \forall s\in S\} 
    \end{align}
 with   $\mu_i(\partial A_{\bx_S})=0$.
\end{lemma}
The sets  $A_{\bx_S}$, as defined above, are often called \emph{rectangular sets}. Of course, $A_{\bx_S}\in\mathcal{B}(\E_d^{(j)})$ for $j=1,\ldots, |S|$. 

\begin{remark} $\mbox{}$
\begin{itemize}
    \item[(a)] If $\bY \in \MRV(\alpha_i, b_i, \mu_i, \E^{(i)}_{d})$ for $i=1,\ldots,d$ and $\mu_{|S|}(A_{\bx_S})>0$ then we able to approximate as $t\to\infty$ the probability
\beao
    \P(\bY\in tA_{\bx_S})\sim \frac{1}{b_{|S|}^{\leftarrow}(t)}\mu_{|S|}(A_{\bx_S}),
\eeao
where $b_{|S|}^{\leftarrow}(t)=\inf\{s\geq 0: b_{|S|}(s)\geq t\}\in \RV_{\alpha_{|S|}}$; see \cite{bingham:goldie:teugels:1989}. 
Hence, if 
$\alpha_i<\alpha_j$ for some $1\leq i<j\leq d$, then for $S_1,S_2\subseteq \mathbb{I}$ with $|S_1|=i$ and $|S_2|=j$,
and $\mu_i(A_{\bx_{S_1}})>0$, $\mu_j(A_{\bx_{S_2}})>0,$ we have
\beao
    \lim_{t\to\infty}\frac{\P(\bY\in tA_{\bx_{S_2}})}{\P(\bY\in tA_{\bx_{S_1}})}=0.
\eeao
This implies that the rate of decay of probabilities differ between the sets in $\E_d^{(i)}$ and those in $\E_d^{(j)}$.
\item[(b)] The advantage of having multivariate regular variation on $\E_d^{(i)}$ is that we not only get estimates of tail probabilities of rectangular sets $tA_{\bx_S}$ with $|S|\geq i$ as $t\to\infty$, but actually for any Borel set bounded away from $\{\by\in \R^d_+: y_{(i)}=0\}$ which are $\mu_i$ continuity sets, e.g., consider the sets
\begin{align}\label{set:Bx}
    B_{\bx,i} = \{\by\in \R^d_+: y_j>x_j \text{ for at least $i$ components of $\bx$}\},
\end{align}
where $\bx=(x_1,\dots,x_d)>\bzero$ and $i\in\mathbb{I}$.  Note that $B_{\bx,i} \in \mathcal{B}(\E_d^{(i)})$ and $B_{\bx,i}$ is not a rectangular set. A special case is
$B_{\bone,i}=\{\by\in \R^d_+:y_{(i)}>1\}$.
\end{itemize}
\end{remark}

\subsection{Tail behavior of a multivariate normal distribution}\label{subsec:multinormal}

In this section, we characterize the tail probability of a multivariate normal random variable $\bZ$ with standard normal marginals, and correlation matrix $\Sigma$. This asymptotic behavior is needed to derive the tail probability approximation for $\P(\bX\in tA_{\bx_S})$ as $t\to\infty$ where $\bX\in \RVGC(\alpha, b,  \Sigma)$.
A direct application of multivariate Mill's ratio of  \citet{savage:1962} for the approximation of Gaussian tail sets  requires that the correlation matrix $\Sigma$ of the Gaussian random vector satisfies $\Sigma^{-1}\bone>\bzero$.
Extensions of Savage's result allow for the approximation of  $\P(\bZ> \boldsymbol{t})$ for large values of $\boldsymbol{t}$ if $\Sigma^{-1}\bone>\bzero$ is not satisfied; see \citep{savage:1962,tong:1989,elnaggar:mukherjea:1999,hashorva:husler:2003,hashorva:2005} for details.
Our result is based on the ideas of  \citet[Theorem 3.4]{hashorva:2005}. First, we state  an auxiliary result on the solution of a quadratic programming problem. This sets up  all notations and background for us to present the  result on the probability of tail sets of a Gaussian random vector. 

\begin{lemma}\label{lem:qp}
Let $\Sigma\in \R^{d\times d}$ be a positive definite correlation matrix. Then the quadratic programming problem 
\begin{align}\label{eq:quadprog}
  \mathcal{P}_{\Sigma^{-1}}: \min_{\{\bz\ge \bone\}} \bz^\top \Sigma^{-1}\bz  
\end{align}
  has a unique solution $\be^*=\be^*(\Sigma)\in\R^d$ such that
    \begin{align}
    \gamma:= \gamma(\Sigma):= \min_{\{\bz\ge \bone\}} \bz^\top \Sigma^{-1} \bz=\be^{*\top} \Sigma^{-1}\be^*>1. 
\end{align}
Moreover, the following hold:
\begin{enumerate}[(a)]
\item If $\Sigma^{-1}\bone \geq \bzero$ then  $\be^*=\bone.$

\item There exists a unique non-empty index set  $I:=I(\Sigma)\subseteq \{1,\ldots, d\}=:\mathbb{I}$ with $J:=J(\Sigma):=\mathbb{I}\setminus I$ such that  the unique solution $\be^*$ is given by
$$\be^*_I=\bone_I, \qquad \text{ and } \qquad\be^*_J=-[\Sigma^{-1}]_{JJ}^{-1}[\Sigma^{-1}]_{JI}\bone_I=\Sigma_{JI}(\Sigma_I)^{-1}\bone_I\geq \bone_J,$$
and moreover, 
$$ \bone_{I}\Sigma_{I}^{-1}\bone_I=\be^{*\top} \Sigma^{-1}\be^*=\gamma>1.$$
 Finally,   $h_i:=h_i(\Sigma):=e_i^\top \Sigma^{-1}_{I}\bone_I>0$ for $i\in I$ and for any $\bz\in\R^d$ the following equality holds: $$\bz^{\top}\Sigma^{-1}\be^*=\bz_I^{\top}\Sigma_{I}^{-1}\bone_I.$$ 
\end{enumerate}
\end{lemma}
\Cref{lem:qp}  is taken from \citet[Proposition~2.1]{hashorva:2005}  and \citet[Proposition 2.5 and Corollary 2.7]{hashorva:husler:2002}. 
Note that, if  $\Sigma^{-1}\bone \geq \bzero$ then $|I|\geq 2$ but it is not necessarily that $|I|=d$. In the next proposition, which is relevant for our main result in \Cref{sec:hrvGaussian}, \Cref{lem:qp} is used to determine the tail probability of a normal random vector $\bZ$ for a particular kind of increasing threshold.

\begin{proposition}\label{prop:gausstail}
    Let $\bZ\sim \Phi_{\Sigma}$ be a normal random vector in $\R^d$, $d\ge 2$, with positive definite correlation matrix $\Sigma\in\R^{d\times d}$.  Define the following quantities:
    \begin{enumerate}[(a)]
    \item The parameters $\gamma=\gamma(\Sigma), I=I(\Sigma)$, $\be^{*}=\be^*(\Sigma)$ and $h_i=h_i(\Sigma)$, $i\in I$ are defined with respect to the solution of the quadratic programming problem $\mathcal{P}_{\Sigma^{-1}}$ as in \Cref{lem:qp}.   
    \item Let  $J:=\mathbb{I}\setminus I$. Define
 $Y_J\sim \mathcal{N}(\bzero_J, \Sigma_J-\Sigma_{JI}\Sigma_{I}^{-1}\Sigma_{IJ})$ if $J\not=\emptyset$ and $Y_J=\bzero$ if $J=\emptyset$. 
 \item Define $\bl_{\infty}:=\lim_{t\to\infty}t(\bone_J-\be_J^*)$, a vector in $\R^{|I|}$ with components either  $0$ or $-\infty$.
  \end{enumerate}
  Let $\bz \in \R^d$ and  $\mathcal{L},u:(0,\infty)\to(0,\infty)$ be measurable functions such that $u(t)\uparrow \infty$ and ${\log(\mathcal{L}(t))}/{u(t)}\to 0,$  as $t\to\infty$.
    Then as $t\to \infty$,
    \begin{align*}
         \P & \left(\bZ >  u(t)\bone + \frac{\bz}{u(t)}  + \frac{\log(\mathcal{L}(t))}{u(t)}\bone +o\left(\frac{1}{u(t)}\right)\right) \nonumber\\
          & \quad\quad\quad\quad\quad =(1+o(1))\Upsilon u(t)^{-|I|} (\mathcal{L}(t))^{-\gamma} \exp\left(-\gamma \frac{u(t)^2}2 - \bz^\top\Sigma^{-1}\be^*\right),
    \end{align*}
where 
\begin{align*}
  \Upsilon:=\Upsilon(\Sigma):=
   \frac{\P(Y_J\geq \bl_\infty)}{(2\pi)^{|I|/2}|\Sigma_{I}|^{1/2}\prod_{i\in I}h_i }.
\end{align*}
\end{proposition}

 \begin{proof}
From \citet[Corollary~3.3]{hashorva:2005} 
we know that for  $\bx\in \R^d$, as $t\to\infty$,
\begin{align*}
        \P(\bZ >u(t)\bone + \bx) =(1+o(1))& \frac{\P(Y_J>\bl_\infty+\bx_J-\Sigma_{JI}\Sigma_I^{-1}\bx_I)}{(2\pi)^{|I|/2}|\Sigma_{I}|^{1/2}\prod_{i\in I}h_i } \\
        & \times u(t)^{-|I|}\exp\left(-\gamma \frac{u(t)^2}2 - u(t)\bx^\top\Sigma^{-1}\be^*-\bx^{\top}\Sigma^{-1}\bx\right).      
    \end{align*}
    Define $\bx^{(t)}:=\frac{\bz}{u(t)}  + \frac{\log(\mathcal{L}(t))}{u(t)}\bone +o\left(\frac{1}{u(t)}\right)$. Therefore, $\lim_{t\to\infty}\bx^{(t)}=\bzero$ and we have $\be_J^*\geq\bone_J$ by \Cref{lem:qp}.
    Thus,
    \beao
    \lim_{t\to\infty}\P\left(Y_J\ge  u(t)(\bone_J-e_J^*)+\bx_J^{(t)}-\Sigma_{JI}\Sigma_I^{-1}\bx_I^{(t)}\right)=\P(Y_J\ge \bl_\infty).
\eeao
Following arguments  analogous to the proof of Theorem~3.1 and Corollary~3.3 
in \citet{hashorva:2005}, we are allowed to replace  $\bx$ by $\bx^{(t)} $ and
$\P(Y_J>\ell_\infty+\bx_J-\Sigma_{JI}\Sigma_I^{-1}\bx_I)$ by $\P(Y_J\ge \bl_\infty)$
to   obtain
\begin{align*}
         \P\Bigg(\bZ > u(t)\bone & + \frac{\bz}{u(t)}  + \frac{\log(\mathcal{L}(t))}{u(t)}\bone +o\left(\frac{1}{u(t)}\right) \Bigg)\\
         &= \P\left(\bZ > u(t)\bone + \bx^{(t)}\right)\\ 
        &=(1+o(1))\Upsilon u(t)^{-|I|}  
        \exp\left(-\gamma \frac{u(t)^2}2 - u(t)\bx^{(t)\top}\Sigma^{-1}\be^*-\bx^{(t)\,\top}\Sigma^{-1}\bx^{(t)}\right) 
\intertext{and since $\bx^{(t)}\to \bzero$ and $u(t)\bx^{(t)}\sim \bz  +\log(\mathcal{L}(t))\bone$ as $t\to\infty$ we have the above to be} 
        &=(1+o(1))\Upsilon u(t)^{-|I|}  
        \exp\left(-\gamma \frac{u(t)^2}2 - \bz^\top\Sigma^{-1}\be^*- \log{\mathcal{L}(t)}\cdot\bone^{\top}\Sigma^{-1}\be^* \right) \\  
       & =(1+o(1))\Upsilon u(t)^{-|I|} (\mathcal{L}(t))^{-\gamma} \exp\left(-\gamma \frac{u(t)^2}2 - \bz^\top\Sigma^{-1}\be^*\right),
    \end{align*}
  where for the last step we used  $\bone^{\top}\Sigma^{-1}\be^*= \bone_I\Sigma_{I}^{-1}\bone_I =\gamma$ from \Cref{lem:qp}.  
 \end{proof}

\section{Multivariate regular variation with Gaussian copula}\label{sec:hrvGaussian} 

In this section, we exhibit multivariate regular variation for $\bX\in \RVGC(\alpha, b,  \Sigma)$ on the different subcones  $\E_d^{(i)}\subset\R_+^d$. We begin by computing  probabilities of such random vectors $\bX$ lying in rectangular sets. This is particularly useful in practice, since many tail risk sets are rectangular in nature. The eventual result on multivariate regular variation of $\bX$ allows an extension of this to more general tail sets.

\begin{theorem}\label{prop:premain}
    Let  $\bX\sim F$ with $\bX\in \RVGC(\alpha, b, \Sigma)$ where $\Sigma$ is positive definite. Fix a non-empty set $S\subseteq \mathbb{I}=\{1,\ldots,d\}$ with $|S|\ge 2$ and let $\gamma_S:=\gamma(\Sigma_S)$, $I_S:=I(\Sigma_S)$, $\be_S^*:=\be^*(\Sigma_S)$, $\Upsilon_S:=\Upsilon(\Sigma_S)$ and $h_{s}^S:=h_s(\Sigma_S), s\in I_S$, be defined as in  \Cref{prop:gausstail} (and in \Cref{lem:qp}). Let 
    \begin{align}\label{def:Axs}
        A_{\bx_S} = \{\by\in \R_+^d: y_s > x_s, \forall s\in S\} 
    \end{align}
   for $\bx_S=(x_s)_{s\in S}$  with $x_s>0, \forall s\in S$ be a rectangular set. If  $\log(\ell(t))=o(\sqrt{\log(t)})$ as $t\to\infty$ where $\overline{F}_1(t)=(t^{\alpha}\ell(t))^{-1}$,  then, as $t\to\infty$,
    \begin{align}\label{eq:mainlimitmrvA}
        \P(\bX\in tA_{\bx_S}) = (1+o(1))\Upsilon_S({2\pi})^{\frac{\gamma_S}2}  (2\alpha\log (t))^{\frac{\gamma_S-|I_S|}{2}}(b^{\leftarrow}(t))^{-\gamma_S}  \prod_{s\in I_S} x_s^{-\alpha h_s^S}.
    \end{align}
\end{theorem}
\begin{remark}
Any distribution $F_1$ with $\ov F_1\in\RV_{-\alpha}$ has the representation $\overline{F}_1(t)=(t^{\alpha}\ell(t))^{-1}$  with $\ell\in \RV_0$.
However, for our results we require the additional assumption $\log(\ell(t))=o(\sqrt{\log(t)})$ as $t\to\infty$. Fortunately, many regularly varying distributions satisfy this property, in particular with $\ell(t)\sim c$ for some constant $c>0$ as, e.g., Pareto, Burr, L\'evy, Student's $t$ and Fr\'ech\'et distributions  (cf. \cite{embrechts:kluppelberg:mikosch:1997, nair:wierman:zwart:2015}).
   
\end{remark}

For the proof of \Cref{prop:premain} we require a relationship between the quantiles of a regularly varying distribution and that of a normal distribution, which we derive next.

\begin{lemma}\label{lem:paretotogaussquantile}
    Let $F_\alpha$ be a distribution function with
    $\ov{F}_{\alpha}(t)=(t^{\alpha}\ell(t))^{-1}$ where $\ell\in \RV_0$ for some $\alpha>0$. 
    Suppose that $F_\alpha$ is strictly increasing and continuous. Fix $x>0$ and define for any $t>0$  $$z_{x}(t)=\Phi^{-1}(F_\alpha(tx))=\overline{\Phi}^{-1}(\overline{F}_\alpha(tx)).$$ 
    Then as $t\to\infty$,
    \begin{align*}
      z_x(t)  =  \sqrt{2\alpha \log (t)} &+   \frac{1}{\sqrt{2\alpha\log (t)}}  \log\left(\frac{\ell(t)}{\sqrt{\log(t)}}\right) \nonumber \\
      &+   \frac{1}{\sqrt{2\alpha\log (t)}} \log\left(\frac{x^{\alpha}}{2\sqrt{\pi\alpha}}\right)+ o\left(\frac1{\sqrt{\log(t)}}\right).
    \end{align*} 
  \end{lemma}

\begin{proof}  From normal quantile approximations (cf. \citet[p. 61, Ex 2.9]{dehaan:ferreira:2006}, \citet[p. 145]{embrechts:kluppelberg:mikosch:1997}) we have as $t\to\infty$,
    \begin{align*}
    \overline{\Phi}^{-1}(t^{-1})& = \Phi^{-1}(1-t^{-1}) = \sqrt{2\log (t)} - \frac{\log(\log(t))+\log(4\pi)}{2\sqrt{2\log t}} + o\left(\frac1{\sqrt{\log(t)}}\right).
    \end{align*}
    Using the representation $\overline{F}_\alpha(t) = \frac{1}{t^{\alpha}\ell(t)}$ for some slowly varying function $\ell$, we obtain 
    \begin{align}
        z_x(t) & =\overline{\Phi}^{-1}\left([(tx)^{\alpha}\ell(tx)]^{-1}\right) \nonumber \\   
               &  = \sqrt{2\alpha\log (tx) + 2\log(\ell(tx))}  - \frac{\log(\log(tx (\ell(tx)^{1/\alpha})))+ \log(\alpha) + \log(4\pi)}{2\sqrt{2\alpha\log (tx) + 2\log(\ell(tx))}} \nonumber \\ & \quad\quad\quad\quad\quad\quad \quad\quad\quad\quad\quad\quad\quad\quad\quad\quad\quad\quad+  o\left(\frac1{\sqrt{2\alpha\log (tx) + 2\log(\ell(tx))}}\right) \nonumber \\
               & =: A_t - B_t + C_t \quad \text{(say)}. \label{quantile}
    \end{align}
 A Taylor series expansion gives
\begin{align*}
    A_t & = \sqrt{2\alpha\log(t)}\left(1+ \frac{\log(x)+ \frac1{\alpha}\log(\ell(tx))}{\log(t)}\right)^{1/2}\\
        & = \sqrt{2\alpha\log(t)}\left[1+\frac{1}{2}\frac{\log(x)}{\log(t)}+ \frac{1}{2\alpha}\frac{\log(\ell(tx))}{\log(t)} +  o\left(\frac1{{\log(t)}}\right)\right] \nonumber\\
        & =  \sqrt{2\alpha \log (t)} +   \frac{\log(x^{\alpha})}{\sqrt{2\alpha\log (t)}} +  \frac{\log(\ell(t))}{\sqrt{2\alpha\log (t)}}  +  o\left(\frac1{\sqrt{\log(t)}}\right).
\end{align*}
Note, in the last step we used $\ell\in\RV_0$ and hence, $\log(\ell(xt))=\log(\ell(x))+o(1)$ as $t\to\infty$.
Similarly, by a Taylor series expansion we have 
\begin{align*}
    B_t & = \frac{\log(\log(t))}{2\sqrt{2\alpha\log t}} + \frac{\log(4\pi\alpha)}{2\sqrt{2\alpha \log t}}  + o\left(\frac1{\sqrt{\log(t)}}\right)\\
       & =  \frac{\log(\sqrt{\log(t)})}{\sqrt{2\alpha\log t}} + \frac{\log(2\sqrt{\pi\alpha})}{\sqrt{2\alpha \log t}}  + o\left(\frac1{\sqrt{\log(t)}}\right).
\end{align*}
Clearly $C_t=o((\log(t))^{-1/2})$ as $t\to\infty$.
Hence, combining $A_t, B_t, C_t$ and \eqref{quantile} gives the result.
\end{proof}

\begin{proof}[Proof of \Cref{prop:premain}] 
For $A_{\bx_S}$ as defined in \eqref{def:Axs}, we have
\begin{align*}
    \P(\bX\in tA_{\bx_S}) & = \P(X_s > tx_s, \forall s \in S) = \P\left(Z_s > z_{x_s}(t), \forall s \in S \right),
\end{align*}
where $z_{x_s}(t)=\ov \Phi^{-1}(\ov F_{s}(t x_s))$, $s\in S$, and $\bZ_S=(Z_s)_{s\in S}\sim \Phi_{\Sigma_S}$. By tail equivalence we have $\ov{F}_s(t)\sim \ov{F}_1(t) = (t^{\alpha}\ell(t))^{-1}$. 
Defining $\bz_S := \left(\log\left({x_s^{\alpha}}/{2\sqrt{\pi\alpha}}\right)\right)_{s\in S}$ and applying \Cref{lem:paretotogaussquantile} we get 
\begin{align} 
    & \P(\bX\in tA_{\bx_S}) \nonumber \\
    & = \P\left(\bZ_S> \sqrt{2\alpha \log (t) } \bone + \frac{\bz_S}{\sqrt{2\alpha \log (t)}} + \frac{\log\left(\frac{\ell(t)}{\sqrt{\log(t)}}\right)}{\sqrt{2\alpha \log (t)}}\bone + o\left(\frac{1}{\sqrt{\log(t)}}\right)\right) \nonumber\\
                  & = \P\left(\bZ_S> u(t) \bone + \frac{\bz_S}{u(t)} + \frac{\log(\mathcal{L}(t))}{u(t)}\bone + o\left(\frac{1}{u(t)}\right)\right), \label{eq4}
\end{align}
where $u(t)=\sqrt{2\alpha\log(t)}$ and $\mathcal{L}(t)= \frac{\ell(t)}{\sqrt{\log(t)}}$. Since by assumption $\log(\ell(t))=o(\sqrt{\log(t)})$ as $t\to\infty$, we have $\lim_{t\to\infty}\log(\mathcal{L}(t))/u(t)=0$.  Using \Cref{prop:gausstail}
and $h_s^S=e_s^\top\Sigma_S^{-1}\be_S^*$, $s\in S$,  we get as $t\to \infty$,
\begin{align*}
    \P(\bX\in tA_{\bx_S})  & = (1+o(1)) \Upsilon_S(\sqrt{2\alpha\log(t)})^{-|I_S|} (\mathcal{L}(t))^{-\gamma_S}\exp\left(-\gamma_S\alpha \log (t)- \bz_S^{\top}\Sigma_S^{-1}e_S^*\right)\\
     & = (1+o(1))\Upsilon_S(2\alpha)^{-\frac{|I_S|}{2}} (\log (t))^{\frac{\gamma_S-|I_S|}{2}} (\ell(t))^{-\gamma_S}\\ & \quad\quad\quad\quad\quad\quad\quad \quad\quad\quad \times t^{-\alpha\gamma_S}  \prod_{s\in I_S} \exp\left(-h_s^S\log\left(\frac{x_s^{\alpha}}{2\sqrt{\pi\alpha}}\right)\right)\\
             & = (1+o(1))\Upsilon_S(2\alpha)^{-\frac{|I_S|}{2}} (\log (t))^{\frac{\gamma_S-|I_S|}{2}} (\ell(t)t^{\alpha})^{-\gamma_S} (2\sqrt{\pi\alpha})^{\gamma_S} \prod_{s\in I_S} x_s^{-\alpha h_s^S}\\
              & = (1+o(1))\Upsilon_S({2\pi})^{\frac{\gamma_S}2}  (2\alpha\log (t))^{\frac{\gamma_S-|I_S|}{2}}(b^{\leftarrow}(t))^{-\gamma_S}  \prod_{s\in I_S} x_s^{-\alpha h_s^S}, 
\end{align*}
where in the penultimate equality we used  that $\sum_{s\in I_S}h_s^S=\bone_{I_S}^{\top}\Sigma^{-1}_{I_S}\bone_{I_S}=\gamma_S$ due to  \Cref{lem:qp}.
\end{proof}
We are now able to present the main theorem of the paper which characterizes regular variation of $\bX\in \RVGC(\alpha,b,  \Sigma)$ on subspaces $\E_d^{(i)}, i=1,\ldots, d$. 

\begin{theorem}\label{thm:hrvgc}
Let $\bX\sim F$ with $\bX\in \RVGC(\alpha, b, \Sigma)$ where $\Sigma$ is positive definite. Suppose that $\log(\ell(t))=o(\sqrt{\log(t)})$ as $t\to\infty$ where $\overline{F}_1(t)=(t^{\alpha}\ell(t))^{-1}$.  For any non-empty set $S\subseteq\mathbb{I}=\{1,\ldots, d\}$, let $\gamma_S:=\gamma(\Sigma_S)$, $I_S:=I(\Sigma_S)$, $\be_S^*:=\be^*(\Sigma_S)$, $\Upsilon_S:=\Upsilon(\Sigma_S)$ and $h_{s}^S:=h_s(\Sigma_S), s\in I_S$ be defined as in  \Cref{prop:gausstail} (and in \Cref{lem:qp}). 
\begin{itemize}
\item[(a)] Let $i=1$. Then $\bX\in \MRV(\alpha,b,\mu_1,\E_d^{(1)})$ with 
\[\mu_1([\bzero,\bx]^c) = \sum_{j=1}^d x_{j}^{-\alpha}, \qquad \forall\; \bx\in \R_+^d. \]
\item[(b)] Let $2\le i \le d$. 
Define 
\begin{align*}
   \mathcal{S}_i&:=\left\{ S\subseteq \mathbb{I}: \; |S|\geq i, \bone_{I_S}\Sigma^{-1}_{I_S}\bone_{I_S}=\min_{\widetilde S\subseteq \mathbb{I}, |\widetilde S|\ge i} \bone_{I_{\widetilde S}}\Sigma^{-1}_{I_{\widetilde S}}\bone_{I_{\widetilde S}}\right\},\\
    I_i&:=\arg \min_{S\in\mathcal{S}_i}|I_{S}|,
\end{align*}
where $I_i$ is not necessarily unique. 
Then $\bX\in \MRV(\alpha_i,b_i,\mu_i,\E_d^{(i)})$ where
 \begin{align*}
    \gamma_i&=\gamma(\Sigma_{I_i})=
    \bone_{I_i}^{\top}\Sigma_{I_i}^{-1}\bone_{I_i}=\min_{S\subseteq \mathbb{I}, |S| \ge i} \min_{\bx_S \ge \bone_S} \bx_S^{\top} \Sigma_S^{-1} \bx_S,\\
   \alpha_i&= 
   \alpha\gamma_i,\\
   b_i^{\leftarrow}(t) &= ({2\pi})^{-\frac{\gamma_i}2} (2\alpha\log (t))^{\frac{|I_i|-\gamma_i}{2}}(b^{\leftarrow}(t))^{\gamma_i},
\end{align*}
 and for a set $A_{\bx_S}=\{\by\in \R^d_+: y_s>x_s, \forall\,s\in S\}$ with $x_s>0, \forall s\in S\subseteq \mathbb{I}$ and $|S|\ge i$, we have
\begin{align}\label{def:muiAxs}
    \mu_i(A_{\bx_S})= \begin{cases}
                   \Upsilon_{S}  \prod_{s\in I_S} x_s^{-\alpha h_s^S}
    & \text{ if }  S\in \mathcal{S}_i \text{ and } |I_S|=|I_i|, \\ 
                   0  &\text{otherwise}. \end{cases} 
\end{align}
\end{itemize}
\end{theorem}

\begin{proof} $\mbox{}$\\
(a) \; For any $\bx\in \R^d_+$, the set $A=[\bzero,\bx]^c$ satisfies $\mu_1(\partial A)=0$ with $\mu_1$ as above. By the inclusion-exclusion principle,  
\[ t\,\P\left(\frac{\bX}{b(t)}\in [\bzero,\bx]^c\right) \le  \sum_{j=1}^d t\,\P(X_j>b(t)x_j) \stackrel{t\to\infty}{\longrightarrow} \sum_{j=1}^d x_{j}^{-\alpha}  \]
by definition. And again, by the inclusion-exclusion principle,  
\begin{align*}
t\,\P\left(\frac{\bX}{b(t)}\in [\bzero,\bx]^c\right) & \ge  \sum_{j=1}^d t\,\P(X_j>b(t)x_j) - \sum_{\genfrac{}{}{0pt}{}{j,k=1}{j\neq k}}^d t\,\P(X_j>b(t)x_j,X_k>b(t)x_k)\\ 
   & \stackrel{t\to\infty}{\longrightarrow} \sum_{j=1}^d x_{j}^{-\alpha},
\end{align*}
since $\P(X_j>b(t)x_j,X_k>b(t)x_k)=o(t)$ for $\rho_{jk}<1$ where $\Sigma=((\rho_{rs}))_{1\leq r,s\leq d}$ (cf. \citet{hua:joe:li:2014}). Therefore $\bX\in \MRV(\alpha,b,\mu_1,\E_d^{(1)})$ follows by \citet[Lemma 6.1]{resnickbook:2007} or \citet[Proposition 2.7]{das:fasen:2023}.

(b) \; Let $i\in \{2,\ldots,d\}$. By \Cref{Lemma 2.3}, it suffices to show that    
\begin{align}\label{lim:mui} \lim_{t\to\infty}t \P\left(\frac{\bX}{b_i(t)}\in A_{\bx_S}\right) =  \lim_{t\to\infty} b_{i}^{\leftarrow}(t) \P\left(\bX\in tA_{\bx_S}\right) = \mu_i(A_{\bx_S})
\end{align}
for any (rectangular) set $A_{\bx_S}=\{\by\in \R_+^d: y_s > x_s, \forall s \in S\}$ where $S\subseteq \mathbb{I}, |S| \ge i$ and $x_s>0, \,\forall\, s\in S$ and $\mu_i(\partial A_{\bx_S})=0$.

Suppose now $A_{\bx_S}$ is such a rectangular set. By the definition of $\mu_i$ from above, for any such set $A_{\bx_S}$, we have $\mu_i(\partial A_{\bx_S})=0$.  From \Cref{prop:premain}, we have 
\begin{align*}
\P(\bX\in tA_{\bx_S})= (1+o(1))\Upsilon_S({2\pi})^{\frac{\gamma_S}2}  (2\alpha\log (t))^{\frac{\gamma_S-|I_S|}{2}}(b^{\leftarrow}(t))^{-\gamma_S}  \prod_{s\in I_S} x_s^{-\alpha h_s^S}.
\end{align*}
 Furthermore, by  definition $\gamma_i=\gamma_{I_i}\le \gamma_S$ for any $S\subseteq \mathbb{I}$ with $|S|\ge i$. \\
 \underline{Case 1:} Suppose $ S\notin \mathcal{S}_i$. Then $\gamma_i<\gamma_S$ and  $\alpha_i=\alpha\gamma_i<\alpha\gamma_S$.  Thus, the function \linebreak $g(t):= b_i^{\leftarrow}(t)\P(\bX\in tA_{\bx_S})\in\RV_{\alpha_i-\alpha\gamma_S}$ converges to $0$  and we have
\begin{align} \label{eq3}
   \lim_{t\to\infty} b_i^{\leftarrow}(t)\P(\bX\in tA_{\bx_S}) =0.
\end{align}
 \underline{Case 2:} Suppose $ S\in\mathcal{S}_i$ and $|I_S|>|I_i|$. Then   $\alpha_i=\alpha\gamma_S$ as well. Although $g\in\RV_0$ we have $g(t)\sim C(2\alpha\log(t))^{\frac{|I_i|-|I_S|}{2}}$ as $t\to\infty$ and hence, again \eqref{eq3} holds.\\
  \underline{Case 3:} 
Suppose $ S\in\mathcal{S}_i$ and $|I_S|=|I_i|$. Then    $\alpha_i=\alpha\gamma_S$ and therefore,
\begin{align*}
   \lim_{t\to\infty} b_i^{\leftarrow}(t)\P(\bX\in tA_{\bx_S}) = \Upsilon_S \prod_{s\in I_S} x_s^{-\alpha h_s^S}.
\end{align*}
Thus, Cases 1-3 show \eqref{lim:mui} for a rectangular set $A_{\bx_S}$ with $|S|\ge i$ and $\mu_i(\partial A_{\bx_S})=0$, hence proving the result.
\end{proof}

\begin{remark} Note that even if we have $S_1, S_2\in\mathcal{S}_i$, it may not necessarily imply that $|I_{S_1}|\not=|I_{S_2}|$. This has an effect on the probability approximation for different tail sets in the same subspace $\E_d^{(i)}$. Consider $\bX\in \R^6$ with $\bX\in \RVGC(\alpha, \Sigma)$ such that   
\begin{align*}
    \Sigma=\left(\begin{array}{c@{\quad}c@{\quad}c@{\quad}c@{\quad}c@{\quad}c}
1 & 0.6 & \sqrt{2}\times 0.6 & 0 & 0 & 0 \\
0.6 & 1 & \sqrt{2}\times 0.6 & 0 & 0 & 0  \\
\sqrt{2}\times 0.6  & \sqrt{2}\times  0.6  & 1 & 0 & 0 & 0  \\
0 & 0 & 0 & 1 & 0.7 & 0.7\\
0 & 0 & 0 & 0.7 & 1 & 0.7\\
0 & 0 & 0 & 0.7 & 0.7 & 1
\end{array}\right).
\end{align*}
For regular variation on $\E_6^{(3)}$ we can check that $\mathcal{S}_3=\{\{1,2\},\{4,5,6\}\}$ with $\gamma_3=1.25$.
But $|\{1,2\}|=2$ and $|\{4,5,6\}|=3$, and hence, $I_3=\{1,2\}$ and $\alpha_3=\alpha\gamma_3$. Interestingly, for fixed $x_1,x_2,x_3,x_4,x_5,x_6>0$  we have both $g_1(t)=\P(X_1>tx_1,X_2>tx_2,X_3>tx_3)\in\RV_{-\alpha_3}$ and 
$g_2(t)=\P(X_4>tx_4,X_5>tx_5,X_6>tx_6)\in\RV_{-\alpha_3}$, yet,
\beao
    \lim_{t\to\infty}\frac{\P(X_4>tx_4,X_5>tx_5,X_6>tx_6)}{\P(X_1>tx_1,X_2>tx_2,X_3>tx_3)}=0,
\eeao
since $g_1$ and $g_2$ have different slowly varying components.
\end{remark}

Besides rectangular sets, \Cref{thm:hrvgc} may be used to compute asymptotic probabilities for sets of the form $B_{\bx,i}$ as defined in \eqref{set:Bx}. We formalize it in the following corollary.
\begin{corollary}\label{corr:mainBx}
Let the assumptions and notations of \Cref{thm:hrvgc} hold. For \linebreak $\bx=(x_1,\ldots,x_d)>\bzero$ and $i\in\mathbb{I}$ define the set
    \begin{align*}
        B_{\bx,i} = \{\by\in \R^d_+: y_j>x_j \text{ for at least $i$ components of $\bx$}\}.
    \end{align*}
    Then  as $t\to\infty$, 
    \begin{align*}
        \P(\bX\in tB_{\bx,i}) = (1+o(1))({2\pi})^{\frac{\gamma_i}2}  (2\alpha\log (t))^{\frac{\gamma_i-|I_i|}{2}}(b^{\leftarrow}(t))^{-\gamma_i} \mu_i(B_{\bx,i}).
    \end{align*}
   If $\alpha_{i+1}>\alpha_i$ then 
    \begin{align*}
         \mu_i(B_{\bx,i})=\sum_{\genfrac{}{}{0pt}{}{S\in\mathcal{S}_i,|S|=i }{|I_S|=|I_i|}}  \Upsilon_{S}  \prod_{s\in I_S} x_s^{-\alpha h_s^S}.
    \end{align*}
\end{corollary}

\begin{proof}
    The result follows from \Cref{thm:hrvgc} and the inclusion-exclusion principle.
\end{proof}

\subsection{Examples}\label{subsec:exmples}
We elaborate the results with a few examples next. In these examples we assume that all pairwise correlation $\rho_{jk}\in (-1,1), \forall j\neq k$ to avoid fully dependent pairs of variables.  We also use the choice of $b_i$ (or $b_i^{\leftarrow}$ to be precise) as described in \Cref{thm:hrvgc}, and do not explicitly specify them in the examples.

\bexam\label{ex:trigauspar} Here we elaborate on the case where $d=3$ with a general matrix  $\Sigma$, although explicit conditions and formulas are not always available without further assumptions on the parameters. 
Suppose $\bX=(X_1,X_2, X_3)\in \R^3_+$ with identical Pareto($\alpha$) marginals and $\bX\in \RVGC(\alpha,b, \Sigma)$ with $b(t)=t^{\frac{1}{\alpha}}$ and $$\Sigma=\begin{pmatrix}
1 & \;\;\;\; \rho_{12} &\;\;\;\; \rho_{13}\\
\rho_{12} & \;\;\;\; 1 & \;\;\;\; \rho_{23}\\
 \rho_{13} & \;\;\;\; \rho_{23} & \;\;\;\; 1
\end{pmatrix},$$ where  $\rho_{jk} \in (-1,1), \forall j\neq k$ and  $|\Sigma|=1-\rho_{12}^2-
\rho_{13}^2-\rho_{23}^2+2\rho_{12}\rho_{13}\rho_{23}>0$. Since the one and two dimensional principal minors of $\Sigma$ are positive, $|\Sigma|>0$ ensures that $\Sigma$ is positive definite. Moreover,
$$ \Sigma^{-1}= \frac{1}{|\Sigma|}\begin{pmatrix}
1-\rho_{23}^2 &\;\;\;\; \rho_{13}\rho_{23} - \rho_{12} &\;\;\;\; \rho_{12}\rho_{23}-\rho_{13}\\
 \rho_{13}\rho_{23} - \rho_{12} &\;\;\;\; 1-\rho_{13}^2 &\;\;\; \rho_{12}\rho_{13}-\rho_{23}\\
\rho_{12}\rho_{23}-\rho_{13} &\;\;\;\; \rho_{12}\rho_{13}-\rho_{23} &\;\;\;\; 1 -\rho_{12}^2
\end{pmatrix}.$$
To start with, let $\Sigma^{-1}\bone >\bzero$  and define $\bh=\Sigma^{-1}\bone$.  Then, using  \Cref{thm:hrvgc}, we have \linebreak $\bX\in \MRV(\alpha_i, b_i, \mu_i, \E_3^{(i)}), i=1,2,3$,
where 
\begin{align*}
     \alpha_1 & =\alpha,\quad \quad \alpha_2 = \alpha\gamma_2= \min_{1\le j\neq k \le 3}\frac{2\alpha}{1+\rho_{jk}} =\frac{2\alpha}{1+\rho_*}, \quad \quad \alpha_3  = \alpha\gamma_3= \alpha\bone^\top\Sigma^{-1}\bone.
    \end{align*}
      With  $\rho^*=\max_{j\neq k}(\rho_{jk})$ and $ \gamma_2= \frac{2}{1+\rho_*}$, for $\bx=(x_1, x_2, x_3)>\bzero$, we have
    \begin{align*} 
       & \mu_1(\{\by\in \R^3_+: y_j>x_j \})  = x_j^{-\alpha}, \quad 1\le j \le3,\\
   & \mu_2(\{\by\in \R^3_+: y_j>x_j, y_k>x_k\})  = \left\{
   \begin{array}{ll}
   \frac{(1+\rho_*)^{3/2}}{2\pi(1-\rho_*)^{1/2}} (x_jx_k)^{-\frac{\alpha}{1+\rho_{*}}}
   & \text{ if }\rho_{jk}=\rho_*,\\
   0 & \text{ if }\rho_{jk}\not=\rho_*, 
   \end{array}\right.\quad 1\le j\neq k \le3,\\
  & \mu_3(\{\by\in \R^3_+: y_1>x_1, y_2>x_2, y_3>x_3\})  = \frac{(h_1h_2h_3)^{-1}}{(2\pi)^{3/2}|\Sigma|^{1/2}}  x_1^{-\alpha h_1}x_2^{-\alpha h_2}x_3^{-\alpha h_3}.
\end{align*}
Hence, for large $t$ we have 
 \begin{align*}
    \P(\bX\in t B_{\bx,1})&=\P(\bX\in t\,[\bzero,\bx]^c) \sim t^{-\alpha}(x_1^{-\alpha} + x_2^{-\alpha} + x_3^{-\alpha}),\\
    \P(\bX\in t B_{\bx,2})&=\P(X_j> tx_j, X_k> tx_k \text{ for some $j\neq k$}) \\
    & \sim { ({2\pi})^{-\frac{\rho^*}{1+\rho_*}}} (2\alpha\log (t))^{-\frac{\rho_*}{1+\rho_*}}t^{-\frac{2\alpha}{1+\rho_*}}\frac{(1+\rho_*)^{3/2}}{(1-\rho_*)^{1/2}} \sum_{\genfrac{}{}{0pt}{}{1\le j< k \le 3}{\rho_{jk}=\rho_*}}(x_jx_k)^{-\frac{\alpha}{1+\rho_{*}}},\\
    \P(\bX\in t B_{\bx,3})&=\P(\bX\in t\,(\bx,\binfty)) \\
   & \sim ({2\pi})^{\frac{\gamma_3-3}{2}} (2\alpha\log (t))^{\frac{\gamma_3-3}{2}}t^{-\alpha\gamma_3}\frac{(h_1h_2h_3)^{-1}}{|\Sigma|^{1/2}}  x_1^{-\alpha h_1}x_2^{-\alpha h_2}x_3^{-\alpha h_3}.
\end{align*}
For a general matrix $\Sigma$, finding a rule for $|\Sigma|>0$ and $\bh=\Sigma^{-1}\bone>\bzero$ is cumbersome. Moreover, 
the case where not all $h_j>0$  requires some attention. Here, we focus on a couple of particular choices of $\rho_{jk}$ to elaborate on these issues.
\begin{enumerate}[(a)]
    \item Suppose $\Sigma=\Sigma_{\rho}$, i.e, $\rho_{jk}=\rho,$  $\forall\; j\neq k$. Then positive definiteness of $\Sigma$ requires \linebreak $|\Sigma|=(1-\rho)^2(1+2\rho)>0$ implying $-\frac12<\rho<1$. Moreover,  for any $\rho<1$ we have $\bh=\Sigma^{-1}\bone>\bzero$. Here $h_i=\frac{1}{1+2\rho}, i=1,2,3$. Using \Cref{thm:hrvgc}  we have \linebreak $\bX\in \MRV(\alpha_i, b_i, \mu_i, \E_3^{(i)})$, $i=1,2,3$,
where 
\begin{align*}
    \alpha_1 & =\alpha,\quad \quad \alpha_2 = \frac{2\alpha}{1+\rho}, \quad \quad \alpha_3  = \frac{3\alpha}{1+2\rho},
    \end{align*}   
and for $\bx=(x_1, x_2, x_3)>\bzero$,
    \begin{eqnarray*} 
       && \mu_1(\{\by\in \R^3_+:\; y_j>x_j \})  = x_j^{-\alpha}, \quad 1\leq j\leq 3,\\
   && \mu_2(\{\by\in \R^3_+:\; y_j>x_j, y_k>x_k\})  = \frac{(1+\rho)^{3/2}}{2\pi(1-\rho)^{1/2}} (x_jx_k)^{-\frac{\alpha}{1+\rho}}, \quad \ 1\le j\neq k \le 3,\\
  && \mu_3(\{\by\in \R^3_+:\; y_j>x_j \text{ for all $j\in \mathbb{I}$}\})  =  \frac{{(1+2\rho)^{5/2}}}{(2\pi)^{3/2}(1-\rho)}  (x_1x_2x_3)^{-\frac{\alpha}{1+2\rho}}.
\end{eqnarray*}
 \item Suppose $\rho_{12}=\rho, \rho_{13}=\rho_{23}=\sqrt{2}\rho$ with $\sqrt{2}|\rho|<1$. Positive definiteness of $\Sigma$ requires $$|\Sigma|= 1-5\rho^2+4\rho^3=(1-\rho)(1+\rho-4\rho^2)=4(1-\rho)\left(\rho-\tfrac{1-\sqrt{17}}8\right)\left(\tfrac{1+\sqrt{17}}8-\rho\right)>0$$ implying $\rho\in \left(\tfrac{1-\sqrt{17}}8, \tfrac{1+\sqrt{17}}8\right)\approx (-0.39, 0.64)$. Assume this holds.  
 \begin{enumerate}[(i)]
     \item Additionally let $\rho<\tfrac{1}{2\sqrt{2}-1} \approx 0.55$, then $\Sigma^{-1}\bone>\bzero$ and 
     \beam \label{eq1}
     h_1=h_2=\frac{1-\sqrt{2}\rho}{1+\rho-4\rho^2} 
     \quad \text{ and }\quad h_3=\frac{1-(2\sqrt{2}-1)\rho}{1+\rho-4\rho^2}.
     \eeam
      Hence, with $\rho\in (\tfrac{1-\sqrt{17}}8,\tfrac{1}{2\sqrt{2}-1})$ a conclusion from \Cref{thm:hrvgc} is that \linebreak $\bX\in \MRV(\alpha_i,b_i, \mu_i, \E_3^{(i)})$,  $i=1,2,3$,
where \begin{align*}
    \alpha_1 & =\alpha,\quad \quad \alpha_2 = \frac{2\alpha}{1+\sqrt{2}\rho}, \quad \quad \alpha_3  = \alpha\frac{3-(4\sqrt{2}-1)\rho}{1+\rho-4\rho^2},
    \end{align*}
and $\mu_1$ is the same as in part (a). For $\bx=(x_1, x_2, x_3)>\bzero$ and  $1\le j<k\le 3$, we obtain
    \begin{align*} 
    \mu_2(\{\by\in \R^3_+:  \;\; y_j>x_j,\, y_k>x_k \})  = \begin{cases} \frac{(1+\sqrt{2}\rho)^{3/2}}{2\pi(1-\sqrt{2}\rho)^{1/2}}(x_1x_3)^{-\frac{\alpha}{1+\sqrt{2}\rho}}, &  (j,k) = (1,3),\\ \frac{(1+\sqrt{2}\rho)^{3/2}}{2\pi(1-\sqrt{2}\rho)^{1/2}}(x_2x_3)^{-\frac{\alpha}{1+\sqrt{2}\rho}}, & (j,k) = (2,3),\\
    0, & \text{otherwise},\end{cases}
      \end{align*}
  and,    
         \begin{align*} 
   \mu_3(\{\by\in \R^3_+:  \;\;& y_1>x_1, y_2>x_2, y_3>x_3 \})   \\ & \quad=  \frac{(1+\rho-4\rho^2)^{5/2}}{(2\pi)^{3/2}(1-\sqrt{2}\rho)^2(1-(2\sqrt{2}-1)\rho)\sqrt{1-\rho}}  x_1^{-\alpha h_1}x_2^{-\alpha h_2}x_3^{-\alpha h_3}.
\end{align*}
  Since $\mu_2(\{\by\in \R^3_+:  \;\; y_1>x_1,\, y_2>x_2 \})=0$, we {are not able to} approximate the tail probability $\P(X_1>tx_1, X_2>tx_2)$  from this, but we can use \Cref{prop:premain} to get as $t\to\infty$,
\[\P(X_1>tx_1, X_2>tx_2) \sim {(2\pi)^{-\frac{\rho}{1+\rho}}} (2\alpha\log (t))^{-\frac{\rho}{1+\rho}}t^{-\frac{2\alpha}{1+\rho}}\frac{(1+\rho)^{3/2}}{(1-\rho)^{1/2}} (x_1x_2)^{-\frac{\alpha}{1+\rho}}.  \]
\item Let $\rho \ge \tfrac{1}{2\sqrt{2}-1}$. Then $(\Sigma^{-1}\bone)_3\le 0$, see \eqref{eq1}.  Hence, with { $\rho\in \left[\tfrac{1}{2\sqrt{2}-1},  \tfrac{1+\sqrt{17}}8\right)$}, from \Cref{thm:hrvgc}, we have $\bX\in \MRV(\alpha_i,{b_i}, \mu_i, \E_3^{(i)}), i=1,2,3$,
where \begin{align*}
    \alpha_1 & =\alpha,\quad \quad \alpha_2 = \frac{2\alpha}{1+\sqrt{2}\rho}, \quad \quad \alpha_3 = \frac{2\alpha}{1+\rho}, 
    \end{align*}  
and  $\mu_1, \mu_2$ are the same as in part (b)(i). However, in part (b)(i), and all previous examples with $\bX\in \R_+^3$, we had  $I_3=\mathbb{I}=\{1,2,3\}$ and hence, $\gamma_3=\bone^\top\Sigma^{-1}\bone$,
whereas, in this example $I_3=\{1,2\}$ resulting in $\gamma_3=\bone_{\{1,2\}}^{\top}\Sigma_{\{1,2\}}^{-1}\bone_{\{1,2\}}$ and hence, a different behavior on $\E_3^{(3)}$ ensues. Moreover,  
for $\bx=(x_1,x_2,x_3)>\bzero$, we have
    \begin{align*} 
   \mu_3(\{\by\in \R^3_+: \, y_j>x_j \text{ for all $j\in\mathbb{I}$}\})  = \frac{(1+\rho)^{3/2}}{4\pi(1-\rho)^{1/2}} (x_1x_2)^{-\frac{\alpha}{1+\rho}}.
  \end{align*}
   Probabilities $\P(\bX\in tB_{\bx,i})$  are also a little bit different from the previous examples hence, we detail them below. For $\bx=(x_1, x_2, x_3)>\bzero$ and large $t$ we have 
 \begin{align*}
   \P(\bX\in t B_{\bx,1})&= \P(\bX\in t\,[\bzero,\bx]^c) \sim t^{-\alpha}(x_1^{-\alpha} + x_2^{-\alpha} + x_3^{-\alpha}),\\
   \P(\bX\in t B_{\bx,2}) &\sim {({2\pi})^{-\frac{\sqrt{2}\rho}{1+\sqrt{2}\rho}}} (2\alpha\log (t))^{-\frac{\sqrt{2}\rho}{1+\sqrt{2}\rho}}t^{-\frac{2\alpha}{1+\sqrt{2}\rho}} \\
   &\quad \quad \times\frac{(1+\sqrt{2}\rho)^{3/2}}{(1-\sqrt{2}\rho)^{1/2}}\left[(x_1x_3)^{-\frac{\alpha}{1+\sqrt{2}\rho}}+(x_2x_3)^{-\frac{\alpha}{1+\sqrt{2}\rho}}\right],\\
  \P(\bX\in t B_{\bx,3})&=\P(\bX\in t\,(\bx,\binfty)) \\ 
  &\sim {({2\pi})^{-\frac{\rho}{1+\rho}}} (2\alpha\log (t))^{-\frac{\rho}{1+\rho}}t^{-\frac{2\alpha}{1+\rho}} \frac{(1+\rho)^{3/2}}{2(1-\rho)^{1/2}} (x_1x_2)^{-\frac{\alpha}{1+\rho}}.
\end{align*}
Using \Cref{prop:premain}, we can give a non-trivial estimate for $\P(X_1>tx_1, X_2>tx_2)$:
 \begin{align*}
       \P(X_1>tx_1, X_2>tx_2) & \sim {({2\pi})^{-\frac{\rho}{1+\rho}}} (2\alpha\log (t))^{-\frac{\rho}{1+\rho}}t^{-\frac{2\alpha}{1+\rho}} \frac{(1+\rho)^{3/2}}{(1-\rho)^{1/2}} (x_1x_2)^{-\frac{\alpha}{1+\rho}}\\
   & \sim  2 \P(\bX\in t\,(\bx,\binfty)).
\end{align*}
The final equivalence is evident from our computations; here we provide an interpretation for this behavior.  Suppose $\bX=(F_\alpha^{\leftarrow}(\Phi(Z_1)),F_\alpha^{\leftarrow}(\Phi(Z_2)),F_\alpha^{\leftarrow}(\Phi(Z_3)))$ where $\bZ\sim \Phi_{\Sigma}$ and $F_\alpha$ denotes the Pareto$(\alpha)$-distribution. Then
\beao
    \P\left(\bX\in t\,(\bx,\binfty)\right)&=&\P(X_3>tx_3|X_1>tx_1,X_2>tx_2)\P(X_1>tx_1,X_2>tx_2)\\
        &=& \P(F_\alpha^{\leftarrow}(\Phi(Z_3))>tx_3|X_1>tx_1,X_2>tx_2) \P(X_1>tx_1,X_2>tx_2)\\
        &\sim &\P(Z_3>0)\P(X_1>tx_1,X_2>tx_2) \quad\quad (t\to\infty)\\
        &=&\frac{1}{2}\P(X_1>tx_1,X_2>tx_2).
\eeao
When $X_1$ and $X_2$ are large, the normal random variable $Z_3=\Phi^{-1}(F_{\alpha}(X_3))$ is large as well if $Z_3$ happens to be positive due to the influence of the high correlation; note that $\rho\in (0.55, 0.64)$ and hence, $\sqrt{2}\rho\in(0.78, 0.91)$. By  symmetry of the normal distributions this occurs with probability $1/2$.

Thus, interestingly, the probability rate of decay for tail sets in $\E_3^{(3)}$ turns out to be the same as that of particular type of tail sets in $\E_3^{(2)}$ which have negligible probability under regular variation in $\E_3^{(2)}$. 
  \end{enumerate}
\end{enumerate}
\eexam
\bexam\label{ex:generaldequi}
Suppose $\bX=(X_1,\ldots, X_d)\in \R^d, d \ge 2,$ and $\bX\in \RVGC(\alpha, b, \Sigma_{\rho})$ where $\Sigma_\rho$ is the $d$-dimensional equi-correlation matrix with $-\tfrac{1}{d-1}<\rho <1$ which makes $\Sigma_\rho$ positive definite. 
Clearly,   $\Sigma_\rho^{-1}\bone>\bzero$.
Also assume that $\log(\ell(t))=o(\sqrt{\log(t)})$ as $t\to\infty$ where $\ov{F}_1(t)=(t^{\alpha}\ell(t))^{-1}$. Using \Cref{thm:hrvgc} we have $\bX\in \MRV(\alpha_i, b_i, \mu_i, \E_d^{(i)}), i=1,\ldots, d$,
where \begin{align*}
    \alpha_i = \frac{i \alpha}{1+(i-1)\rho}, \quad i=1,\ldots, d.
    \end{align*}
For any non-empty $S\subseteq \mathbb{I}$ with $|S|=i$, we have $|\Sigma_S|= (1-\rho)^{i-1}(1+(i-1)\rho)$.  Now, for
$\bx=(x_1,\ldots,x_d)>\bzero$ we obtain
    \begin{align*} 
       & \mu_1(\{\by\in \R^d_+: y_j>x_j \text{ for some $j$}\})= \mu_1(B_{\bx,1})  = \sum_{j=1}^d x_j^{-\alpha} ,\intertext{and for $2\le i \le d$, with $S\subseteq \mathbb{I}, |S|=i$, we have}
     \mu_i(A_{\bx_S})&=\mu_i(\{\by\in \R^d_+: y_j>x_j, \, \forall\, j \in S\})  = \frac{(1+(i-1)\rho)^{i-1/2}}{(2\pi)^{i/2}(1-\rho)^{(i-1)/2}} \prod_{j\in S} x_j^{-\frac{\alpha}{1+(i-1)\rho}}, \quad \text{ and, }\\
      \mu_i(B_{\bx,i})  &= \frac{(1+(i-1)\rho)^{i-1/2}}{(2\pi)^{i/2}(1-\rho)^{(i-1)/2}} \sum_{\genfrac{}{}{0pt}{}{S\subseteq \mathbb{I}}{|S|=i}} \prod_{j\in S} x_j^{-\frac{\alpha}{1+(i-1)\rho}}.
\end{align*}
Now, probabilities $\P(\bX\in tA_{\bx_S})$and $\P(\bX\in t B_{\bx,i})$  can again be approximated as in the previous examples.
\eexam
 Although \Cref{thm:hrvgc} gives a general formula, the related constants are not easy to compute in a closed form  which provide the precise tail rates, especially as the dimension $d$ increases.  Of course, if we know for all $S\subseteq \mathbb{I}$ the numerical value {of the quadratic programming problem} $\mathcal{P}_{\Sigma^{-1}_S}$ these rates can be computed; however, this might be up to $2^d$ quadratic programming problems.

\subsection{Comparing heavy tails to light tails for Gaussian copula} \label{sec:htvslt}
Gaussian copula have been popular in capturing dependence for a whole range of applications where the marginal distributions are modeled based on expert knowledge or empirical evidence and may range from light-tailed like normal or exponential to more heavy-tailed like Pareto or log-normal distribution. Interestingly, though the asymptotic joint tail behavior may vary significantly depending on whether marginals are heavy-tailed or light-tailed. In the following result, we study two bivariate distributions with the same Gaussian copula and elaborate how the joint probabilities differ in each case.

\begin{proposition}\label{prop:NormvsHt2D}
Let $\bX=(X_1, X_2)\in \RVGC(\alpha,\Sigma_{\rho})$ with identical Pareto($\alpha$) marginals $F_\alpha$, 
$\alpha>0$, and let $\bZ=(Z_1,Z_2)\sim \Phi_{\Sigma_\rho}$
with  $\rho\in (-1,1)$.  
Then the following holds for $x_1, x_2>0$ and $x_{max}=\max(x_1,x_2)$.
\begin{enumerate}
    \item[(a)] If $\rho< \min(x_1/x_2, x_2/x_1)$, then $\P(Z_1>tx_1, Z_2 >tx_2) \sim o(\P(Z_1>t x_{max}))$ as $t\to \infty$ and and additionally if $\rho>0$ we also have
    \begin{align}\label{eq:gaussnormal1}
    \P\left(Z_1> t x_1, Z_2>t x_2\right) \sim \frac{(1-\rho^2)^{3/2}}{(x_1-\rho x_1)(x_2-\rho x_1)}\frac{\varphi(tx_\rho )}{\sqrt{2\pi}t^2},\qquad t\to\infty,
    \end{align}
    where $x_\rho^2:=(x_1^2-2\rho x_1x_2+x_2^2)/(1-\rho^2)$.\\
     If $\rho \ge \min(x_1/x_2, x_2/x_1)$, then $\P(Z_1>tx_1, Z_2 >tx_2) \sim O(\P(Z_1>tx_{max}))$ as $t\to \infty$ and
    \begin{align}\label{eq:gaussnormal2}
    \P\left(Z_1> t x_1, Z_2>t x_2\right) \sim C_{\rho}\frac{\varphi(tx_{max})}{tx_{max}}, \qquad t\to\infty,
    \end{align}
    where $C_{\rho}=\left\{\begin{array}{ll}
        1 & \text{ if } \rho>\min(x_1/x_2,x_2/x_1),\\ 
        \frac{1}{2} & \text{ if  } \rho= \min(x_1/x_2,x_2/x_1).
    \end{array}\right.$
    \item[(b)] For any $\rho\in (-1,1)$ have $\P(X_1>tx_1, X_2 >tx_2) \sim o(\P(X_1>tx_{max}))$ as $t\to \infty$ and
    \begin{align}\label{eq:gausspareto}
     \P(X_1>tx_1, X_2 >tx_2) \sim t^{-\frac{2\alpha}{1+\rho}}(2\alpha\log(t))^{-\frac{\rho}{1+\rho}} {(2\pi)^{-\frac{\rho}{1+\rho}} }
     \frac{(1+\rho)^{3/2}}{(1-\rho)^{1/2}} (x_1x_2)^{-\frac{\alpha}{1+\rho}}, t\to\infty.
    \end{align}
       
\end{enumerate}
\end{proposition}
\bproof {\color{white}{nthing}}
\begin{enumerate}[(a)]
\item These results can be derived from \citet[Lemma 2.1]{elnaggar:mukherjea:1999} with a simple change of variable; also see \citet[Example 3]{hashorva:husler:2003}.
    \item This is a consequence of  \Cref{thm:hrvgc} when $d=2$, see \Cref{ex:generaldequi} with $d=2$ as well.    
\end{enumerate}
\eproof

\begin{remark} $\mbox{}$ \label{Remark:3.9}
\begin{itemize}
\item[(i)] Asymptotic independence in the Gaussian setup is equivalent to $\P(Z_1>t,Z_2>t)=o(\P(Z_2>t))$ as $t\to\infty$. Nevertheless, from \Cref{prop:NormvsHt2D}(a) we have that for any $x_1\in(0,\rho]$,
\beao
    \P(Z_1>x_1 t,Z_2>t)=O(\P(Z_2>t)), \quad t\to\infty.
\eeao
Hence, we need to be careful in interpreting asymptotic independence for this case, since joint tail crossings may still happen with the same rate as the marginal tail crossings.
\item[(ii)] Curiously,  in the Gaussian setting  the rate of decrease of the tail probabilities depend on the chosen $\bx$ resulting in, e.g.,
       \beao
            \lim_{t\to\infty}\frac{\P(Z_1>t,Z_2>t)}{\P(Z_1>x_1 t,Z_2>t)}=0, \quad \forall\;0<x_1\leq \rho.
       \eeao
In contrast, in the heavy-tailed setting (with Gaussian copula) the different tail probabilities are equivalent and 
       \beao
            \lim_{t\to\infty}\frac{\P(X_1>t,X_2>t)}{\P(X_1>x_1 t,X_2>x_2t)}={(x_1x_2)}^{\frac{\alpha}{1+\rho}}>0, \qquad \forall\, x_1,x_2>0.
       \eeao
An explanation for the difference between the two phenomena is the following:
\begin{enumerate}[(a)]
       \item Gaussian model: The assumption $\rho<\min(x_1/x_2,x_2/x_1)$ is equivalent to 
       \beam \label{4.1}
       \Sigma^{-1}\bx=\frac{1}{1-\rho^2}\left(\genfrac{}{}{0pt}{}{x_1-\rho x_2}{x_2-\rho x_1}\right)>\bzero
       \eeam
       with $\bx=(x_1,x_2)\in\R_+^2$. If \eqref{4.1} is satisfied
       then the multivariate version of Mill's ratio (see \citet{savage:1962}) can be applied. If \eqref{4.1}  is not satisfied the multivariate version of Mill's ratio is not directly applicable. 
       \item Heavy-tailed model:  We need not differentiate between the cases $\rho<\min(x_1/x_2,x_2/x_1)$  and $\rho\ge\min(x_1/x_2,x_2/x_1)$ anymore. Indeed,
       \beao
            \P(X_1>tx_1,X_2>tx_2)=\P(Z_1>z_{x_1}(t),Z_2>z_{x_2}(t))
       \eeao
       with 
       $            z_x(t)  =  \overline{\Phi}^{-1}(\overline{F}_\alpha(tx))$
       and due to \Cref{lem:paretotogaussquantile}, 
       \beam
            \frac{1}{\sqrt{2\alpha \log (t)}}\Sigma^{-1}\left(\genfrac{}{}{0pt}{}{z_{x_1}(t)}{z_{x_2}(t)}\right)
            \stackrel{t\to\infty}{\longrightarrow}
            \Sigma^{-1}\bone, 
       \eeam
       allowing a direct application of the multivariate Mill's ratio if $\Sigma^{-1}\bone>\bzero$; in contrast to \eqref{4.1} in the Gaussian case. Of course, in dimension $d=2$ the condition       $\Sigma^{-1}\bone>\bzero$ is  equivalent to $\rho<1$ and is clearly independent of $\bx$. 
\end{enumerate}
\end{itemize}
\end{remark}
\begin{remark}\label{rem:htvsltinhigherdim}
    \Cref{prop:NormvsHt2D} elaborates on a dichotomy of behavior of Gaussian copula under different distribution tails in dimension $d=2$, where asymptotic (tail) independence has been classically defined.  A natural extension of asymptotic (tail) independence to a higher dimensional random vector $\bY\in \R^d$ is that the following holds as $t\to\infty$: $$\P(Y_j>t, \forall j\in \mathbb{I}) = o\left(\min_{\genfrac{}{}{0pt}{}{S \subseteq \mathbb{I}, }{0<|S|<d}} \P(Y_j>t, \forall j\in S)\right).$$
    For this definition of asymptotic independence, the result in \Cref{prop:NormvsHt2D} can be extended  to any dimension $d\geq 2$ as well by comparing \Cref{thm:hrvgc} with \citet[Theorem 4.1]{hashorva:husler:2003}; we have refrained from this at present, opting for brevity and clarity in explaining the idea in two dimensions.
    
\end{remark}

\section{Simulation study}\label{sec:simulation}
In this section, we use simulated data examples to understand the behavior of tail parameters in multivariate heavy-tailed distributions.

\subsection{Comparing tail indices in different dimensions}\label{subsec:ex1} 
Consider the random vector \linebreak $\bX=(X_1, X_2, X_3)$  having identical Pareto($\alpha$) marginal distributions $F_\alpha$ with $\alpha=2$ and dependence given by the Gaussian copula $C_{\Sigma}$ with correlation matrix 
\begin{align}\label{corr3:noneq}
    \Sigma=\begin{pmatrix}
1 & \rho & \sqrt{2}\rho\\
\rho & 1 & \sqrt{2}\rho \\
\sqrt{2}\rho  & \sqrt{2}\rho  & 1
\end{pmatrix};
\end{align}
 cf. \Cref{ex:trigauspar}(b), i.e., $\bX\in \RVGC(\alpha, b, \Sigma)$ with $b(t)=t^{\frac{1}{\alpha}}$. Consider computing probabilities for the events $tA$ when $t$ is large, for various $A\subset \E_3^{(i)}\subset \R^3_+, i=1,2, 3.$ For $1\le j \neq k \le 3$, consider the following sets.
\begin{itemize}
    \item In $\E_3^{(1)}$:
    \begin{eqnarray*}
    \begin{array}{rll}
         A_j & =\{\by\in\R_+^3:\,y_j>1\} &\quad \text{ with }  \quad  \P(\bX\in t A_j)=\P(X_j>t),\\
         A_{\max}  &=\{\by\in\R_+^3:\,\max_j y_j>1\} &\quad \text{ with } \quad \P(\bX\in tA_{\max})=\P(\max_jX_j>t).
    \end{array}
    \end{eqnarray*}
      \item In $\E_3^{(2)}$:
    \begin{align*}
    \begin{array}{rll}
         A_{jk}&=\{\by\in\R_+^3:\, y_j>1,y_k>1\} & \quad \text{ with } \quad  \P(\bX\in tA_{jk})=\P(X_j>t,X_k>t),\\
         A_{(2)}&=\{\by\in\R_+^3:\,y_{(2)}>1\} &\quad \text{ with } \quad \P(\bX\in tA_{(2)})=\P(X_{(2)}>t).
    \end{array}
    \end{align*}
      \item In $\E_3^{(1)}$: 
      \begin{align*}
      A_{\text{min}}=\{\by\in\R_+^3:\,y_j>1\text{ for all }j\} \hspace*{0.1cm}
   \text{  with }\quad  \P(\bX\in tA_{\text{min}})=\P(\min_j X_j>t).
    \end{align*}
\end{itemize}
Note that $y_{(2)}$ denotes the second largest value of $y_1,y_2,y_3$. From \Cref{ex:trigauspar}, we know that $\P(\bX\in tA)=t^{-\alpha_A}\ell_A(t)$  for some $\alpha_A>0, \ell_A\in\RV_0$ (here $\ell_A$ includes both the measure $\mu_*(A)$ of the set $A$ in terms of a limit measure $\mu_*$ and the slowly varying function in \Cref{thm:hrvgc}).  Here, we may estimate the tail parameter $\alpha_A$ using the Hill estimator \citep{hill:1975, resnickbook:2007} for data pertaining to the relevant quantities (see \citep{mitra:resnick:2011hrv,das:mitra:resnick:2013}); for example, the tail index $\alpha_{A_{\min}}$ of  $\P(\bX\in tA_{\min})=\P(\min_{j\in\in\mathbb{I}} X_j >t)=t^{-\alpha_{A_{\min}}}\ell_{A_{\min}}(t)$ can be estimated by the Hill estimator of data coming from iid random variables following the distribution of $\min_{j\in\mathbb{I}} X_j$.
We generate $n=20,000$ samples each from the distribution of $\bX$, using two  choices of $\rho$: $\rho=0.2$, and $\rho=0.6$; then we estimate the tail parameters $\alpha_A$ (decay rates) for the tail probabilities mentioned above. 

\begin{figure}[ht]
    \centering
 \includegraphics[width=\linewidth]{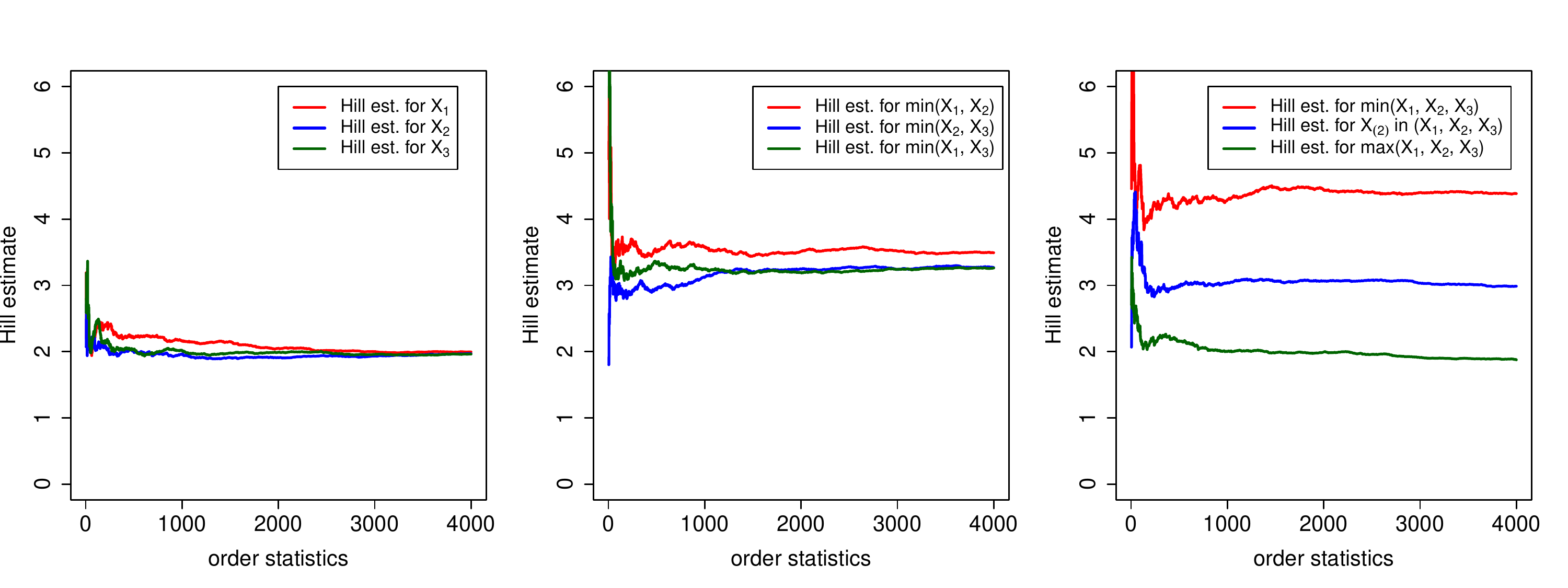} 
 \caption{Hill estimates when $X_1, X_2, X_3 \sim$ Pareto(2) having Gaussian copula dependence $C_\Sigma$ with $\Sigma$ as in \eqref{corr3:noneq} and $\rho=0.2$. (1) Left plot: Hill estimates of the tail indices of $X_1, X_2$ and $X_3$, respectively. (2) Center plot: Hill estimates of the tail indices of $\min(X_1, X_2), \min(X_2, X_3)$ and $\min(X_1,X_3)$, respectively. (3) Right plot:  Hill estimates of the tail index of $\min(X_1, X_2, X_3)$, $X_{(2)}$  and $\max(X_1, X_2, X_3)$.} 
  \label{fig:hrvcorr02}
\end{figure}

The different Hill estimates of tail indices when $\rho=0.2$ are plotted in \Cref{fig:hrvcorr02}. Following tradition, Hill estimates are plotted for an increasing number of order statistics. Stability in the plots supports that the probability tails are regularly varying (heavy-tailed) and the zone of stability approximates the tail index value. For the sets in $\E_3^{(1)}$, the Hill estimates of tail indices $\alpha_{A_j}$ ($j=1,2,3$) for $X_j$ ($j=1,2,3$) and $\alpha_{A_{\max}}$ for $\max_{j}X_j$ are plotted in the left plot and with the green line in the right plot, respectively. All the Hill plots indicate that the relevant tail indices are near $2$ supporting the assumption that the marginal distributions are Pareto with tail parameter $\alpha=2$ which is as well the index of regular variation of $\alpha_1=\alpha=2$ in $\E_3^{(1)}$; the same holds true when $\rho=0.6$ as well, see \Cref{fig:hrvcorr06}. The Hill estimates for the parameters $\alpha_{A_{12}},\alpha_{A_{13}},\alpha_{A_{23}}$ and $\alpha_{A_{(2)}}$ for tail sets in $\E_3^{(2)}$ are plotted in the center plot and the blue line on the right plot, respectively. The plots indicate that the tail indices are all close to 3 with the estimate for the index $\alpha_{A_{12}}$ being higher than those of  $\alpha_{A_{13}}$, $\alpha_{A_{23}}$ as well as $\alpha_{A_{(2)}}$. From \Cref{ex:trigauspar} (b)(i) with $\rho<0.55$, we know that the tail index  for the former is given by $\alpha_{A_{12}}=\frac{2\alpha}{1+\rho}=3.33$ and the latter three by $\alpha_2=\alpha_{A_{13}}=\alpha_{A_{23}}=\alpha_{A_{(2)}}=\frac{2\alpha}{1+\sqrt{2}\rho}=3.12$, supporting the simulation results. The Hill estimates for the tail index  $\alpha_{A_{\min}}$, where $A_{\min}$ is an $\E_3^{(3)}$ set, is given by the red line in the right plot; it is higher than any of the other estimates and is closer to 4. Its true value from \Cref{ex:trigauspar} (b)(i) is $\alpha_{A_{\min}}=\alpha_3=\frac{\alpha(3-(4\sqrt{2}-1)\rho)}{1+\rho-4\rho^2}=3.98$, which again supports the simulations.

\begin{figure}[ht]
    \centering
 \includegraphics[width=\linewidth]{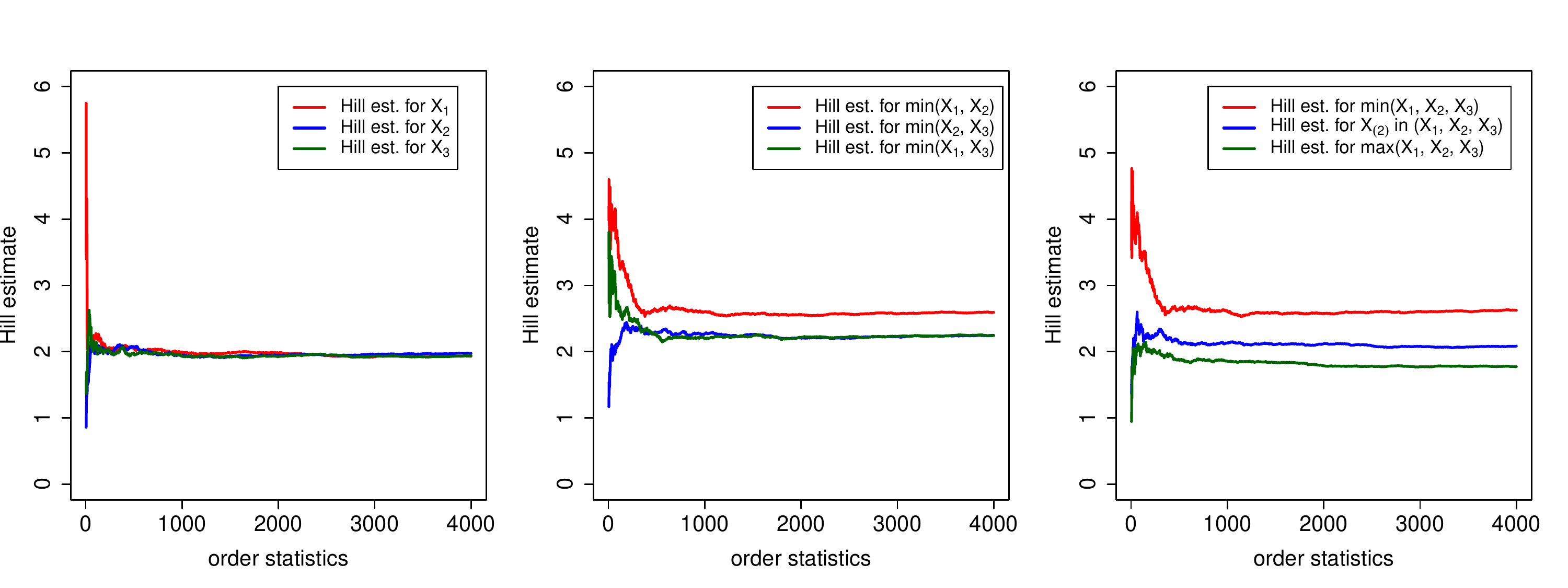}  
 \caption{Hill estimates when $X_1, X_2, X_3 \sim$ Pareto(2) having Gaussian copula dependence $C_\Sigma$ with $\Sigma$ as in \eqref{corr3:noneq} and $\rho=0.6$. (1) Left plot: Hill estimates of the tail indices of $X_1, X_2$ and $X_3$, respectively. (2) Center plot: Hill estimates of the tail indices of $\min(X_1, X_2), \min(X_2, X_3)$ and $\min(X_1,X_3)$, respectively. (3) Right plot:  Hill estimates of the tail index of $\min(X_1, X_2, X_3)$, $X_{(2)}$  and $\max(X_1, X_2, X_3)$.} 
  \label{fig:hrvcorr06}
\end{figure}

In contrast, we plot  in \Cref{fig:hrvcorr06}, the relevant Hill estimates under the same model but now with $\rho=0.6$. This model follows \Cref{ex:trigauspar}(b)(ii) where $\rho\ge 0.55$. The pattern for the behavior of $\E_3^{(1)}$ and $\E_3^{(2)}$ tail sets remain the same except that   $\alpha_{A_{12}}=\frac{2\alpha}{1+\rho}=2.5$ and that  $\alpha_2=\alpha_{A_{(2)}}=\frac{2\alpha}{1+\sqrt{2}\rho}=2.16$, with $\alpha_{A_{12}}>\alpha_2$ which is also supported by the simulation results. Interestingly, the plot for the estimates of the tail index $\alpha_{A_{\min}}$ for the $\E_3^{(3)}$ tail set $A_{\min}$  given by the red line on the right plot is less than 3 and in fact, closer to the Hill plot of $\min(X_1,X_2)$ for $\alpha_{A_{12}}$. This is justified by the fact that here we have 
$\alpha_3=\frac{2\alpha}{1+\rho}=2.5 =\alpha_{A_{12}}$.  It is interesting to note that the tail index for regular variation 
on $\E_3^{(3)}$  is the same as that of the particular set $A_{12}$ in $\E_3^{(2)}$. Yet  $\alpha_{A_{12}}=2.5$ is greater than the regular variation index $\alpha_2=2.16$ on $\E_3^{(2)}$.

In summary, the simulation results exhibit the contrasting behavior of the tail indices for probabilities of tail sets under a Gaussian copula with different parameter values but the same Pareto marginal distribution confirming the theoretical results of \Cref{sec:hrvGaussian}.


\subsection{Comparing normal tails to Pareto tails}\label{subsec:ex2}
Let $\bZ=(Z_1,Z_2)$ have identical standard normal marginal distributions and $\bX=(X_1, X_2)$ have identical Pareto($\alpha$) marginal distributions with $\alpha=2$; moreover, let both $\bX$ and $\bZ$ have the same copula dependence given by the Gaussian copula $C_{\rho}$ with correlation $\rho=2/3.$
\begin{figure}[t]
    \centering
 \includegraphics[width=0.4\linewidth]{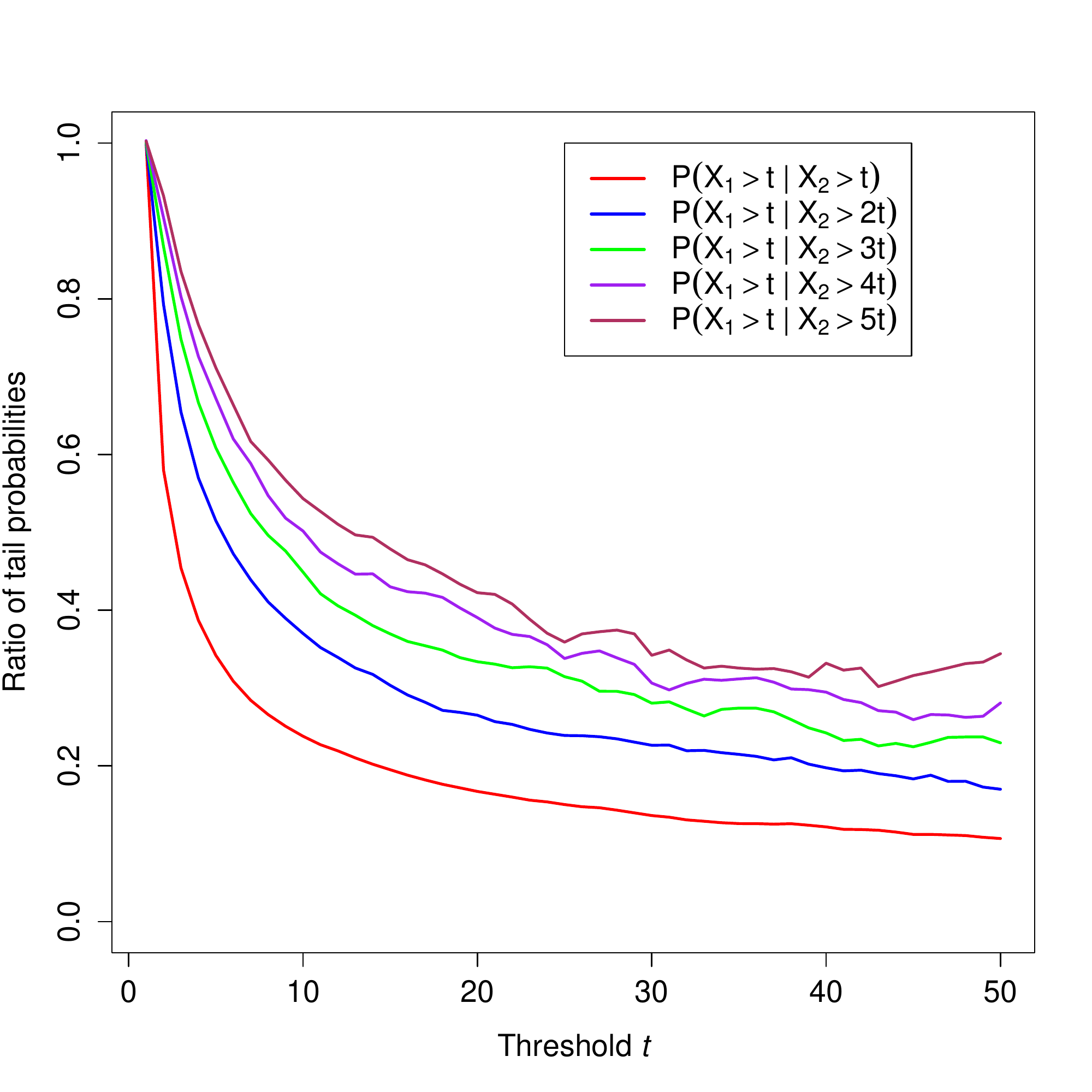}  \quad\quad\quad
 \includegraphics[width=0.4\linewidth]{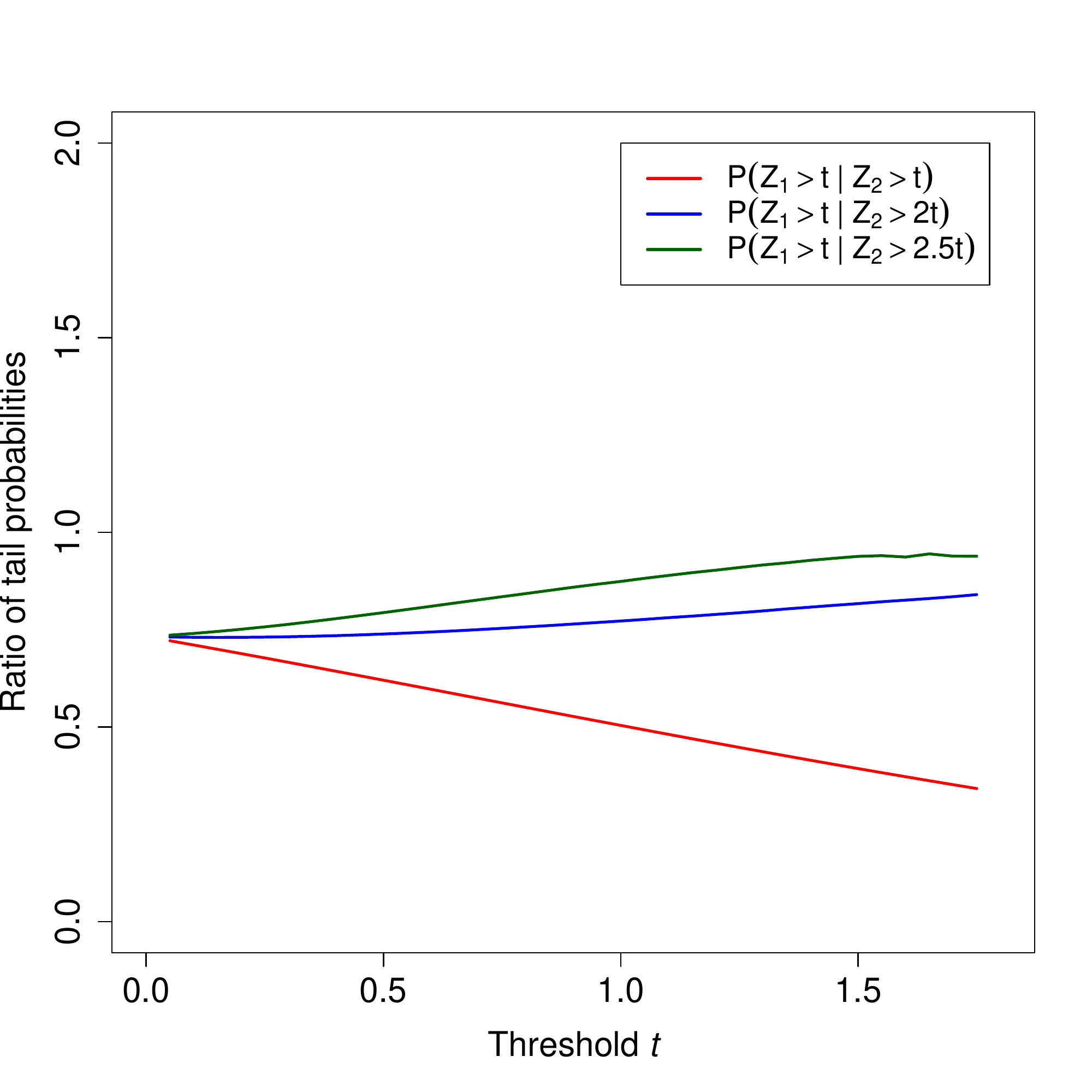}
 \caption{Tail probability comparisons under bivariate Gaussian copula $C_\rho$ with $\rho=2/3$. (1) Left plot: plot of $\P(X_1>t|X_2>\kappa t)$ for $ \kappa=1,\ldots,5$, $t\in (0,50)$ when $X_1, X_2 \sim \text{Pareto}(2)$. (2) Right plot: plot of $\P(Z_1>t|Z_2>\kappa t)$ for $\kappa=1,2,2.5$, $t\in (0,2)$ when $Z_1, Z_2 \sim \mathcal{N}(0,1)$.} 
  \label{fig:asyind}
\end{figure}

 We generate $n=1,000,000$ samples from each distribution $\bZ$ and $\bX$ respectively, and then empirically approximate $\P(X_1>t|  X_2>\kappa t)$ with $\kappa=1,\ldots,5$ and $\P(Z_1>t|  Z_2>\kappa t)$ with $\kappa=1,2,2.5$ for a range of values of $t$. We do this $m=1,000$ times and plot the average tail probabilities in \Cref{fig:asyind}. For the heavy-tailed vector $\bX$, we observe that \linebreak $\P(X_1>t|X_2>\kappa t)$ tends to $0$ for all value of $\kappa$ as $t\to\infty$. Hence, here we may expect  $$\P(X_1> t,X_2>\kappa t) = o\left(\P(X_2>\kappa t)\right), \quad t\to\infty,$$ for any $\kappa\ge 1$, indicating that the phenomenon of asymptotic independence \eqref{eq:asyind} holds more broadly. This is corroborated by \Cref{prop:NormvsHt2D}(b) as well.
 
 On the other hand, for the Gaussian vector,  the probability $\P(Z_1>t|Z_2> t)$ tends to $0$ (as $t\to\infty$) as expected by asymptotic independence, yet $\P(Z_1>t|Z_2> 2 t)$ and \linebreak $\P(Z_1>t|Z_2> 2.5 t)$ seem to stabilize  to some non-zero constant value.  Thus, in contrast to the heavy-tailed case, for Gaussian vectors, under certain circumstances we may expect  $$\P(Z_1>t,Z_2>\kappa t) \sim C \P(Z_2>\kappa t), \quad t\to\infty,$$
for some constant $C>0$. The reason is explained in \Cref{prop:NormvsHt2D}(a) and \Cref{Remark:3.9}: for $\rho=2/3$, with $(x_1,x_2)=(1,1)$, we have $\rho<\min(x_1/x_2,x_2/x_1)$, hence, \eqref{eq:gaussnormal1} holds justifying the red line decreasing to zero, and with  $(x_1,x_2)=(1,2) $ and $ (1,3)$, respectively, we have $\rho > \min(x_1/x_2,x_2/x_1)$, hence \eqref{eq:gaussnormal2} holds justifying the stable behavior of the blue and green lines in the right plot of \Cref{fig:asyind}.

\section{Conclusion}\label{sec:concl}
In this paper, we provide precise asymptotic probabilities for different tail sets of regularly varying distributions with Gaussian copula. 
Our key findings can be summarised as follows.
\begin{enumerate}[(i)]
    \item Multivariate distributions with tail equivalent regularly varying marginals (with minor regularity conditions) and a Gaussian copula dependence admit multivariate regular variation on cones $\E_d^{(i)}\subseteq \R^d_+$ for all $i\in \mathbb{I}$ (cf. \Cref{thm:hrvgc}).
    \item While computing probabilities of rectangular tail sets, which are quite important in practical applications, knowing the regular variation behavior on the relevant space $\E_d^{(i)}$ may still provide a negligible estimate for the tail probabilities; see \Cref{ex:trigauspar}(b). Hence, \Cref{prop:premain}, which provides the tail decay rates for such sets turns out to be quite useful; additionally, \Cref{corr:mainBx} gives tail asymptotics for another relevant tail set.
    For computing tail probabilities of general non-rectangular sets in $\E_d^{(i)}$, we still need to use \Cref{thm:hrvgc}.
    \item The joint tail behavior for light-vs-heavy-tailed distributions under Gaussian copula are structurally quite different. Classical asymptotic independence given by \linebreak $\P(X_1>t| X_2>t) \to 0 $ (as $t\to\infty$) completely characterizes the joint tail behavior when $(X_1,X_2)$ are marginally heavy-tailed. In contrast, if $(Z_1,Z_2)$ has Gaussian marginals,  then $\P(Z_1>t| Z_2>t) \to 0 $ (as $t\to \infty$) only gives a part of the story, as we find in \Cref{sec:htvslt}.
\end{enumerate}
These results are not only useful for risk managers interested in computing tail probabilities for various extreme tail events using a widely popular model, they also enable the computation of certain multivariate risk measures of interest, see \cite{prekopa:2012, cousin:bernardino:2013, cousin:bernardino:2014}.
A few interesting topics in this context still remain to be explored. In all our examples where $\bX\in \RVGC(\alpha,\Sigma)$ with a positive-definite $\Sigma$ and (as a consequence) $\bX\in \MRV(\alpha_i, \E_d^{(i)})$ for all $i\in \mathbb{I}$, we observed that $\alpha_1<\ldots<\alpha_d$ everywhere. But is this universally true for any valid choice of $\Sigma$? Moreover,  the tail asymptotics obtained in \Cref{prop:premain} also hint at statistical testing methods for the Gaussian copula assumption by comparing estimates of tail index parameters for different tail sets. {Finally, the case where the tails of the marginal distributions are regularly varying with different tail index parameters can also be derived under certain conditions from our results.} We leave these questions for the interested researchers to pursue.  

\bibliographystyle{imsart-nameyear}
\bibliography{bibfilenew}



\end{document}